\documentclass[10pt, conference, letterpaper]{IEEEtran}
\pdfoutput=1
\usepackage{amsmath}
\usepackage{mathrsfs}
\usepackage{amssymb}
\usepackage{tikz}
\usetikzlibrary{arrows}

\usepackage{url}
\usepackage{amsmath}
\usepackage{amsthm}
\usepackage{algorithmicx, algorithm ,algpseudocode}
\usepackage{verbatim}
\usepackage{graphicx}
\usepackage{subfigure}

\newtheorem{theorem}{Theorem}
\newtheorem{lemma}{Lemma}

\graphicspath{{figs/}}

\begin{document}

\theoremstyle{definition}
\newtheorem{myprop}{Property}
\newtheorem{mydef}{Definition}

\title{Multiplex Influence Maximization in Online Social Networks with Heterogeneous Diffusion Models}
\author{}
\author{Alan Kuhnle, Md Abdul Alim, Xiang Li, Huiling Zhang, and My T. Thai\\
Department of Computer \& Information Science \& Engineering\\
University of Florida\\
Gainesville, Florida, USA\\
Email: \{kuhnle,abdul.alim,lixiang,huilingzhang,mythai\}@ufl.edu }
\maketitle
\begin{abstract}\label{label:abstract}
Motivated by online social networks that are linked together
through overlapping users,
we study the influence maximization problem on a multiplex,
with each layer endowed with its own model of influence diffusion.
This problem is a novel version of the influence maximization problem 
that necessitates new analysis incorporating the type of
propagation on each layer of the multiplex.
We identify a new property, generalized deterministic submodular, 
which when satisfied by the propagation
in each layer, ensures that the propagation on the multiplex overall
is submodular -- for this case,
we formulate ISF, the greedy algorithm with
approximation ratio $(1-1/e)$.
Since the size of a multiplex comprising multiple OSNs may 
encompass billions of users, we
formulate an algorithm KSN that runs on each layer
of the multiplex in parallel.
KSN takes an $\alpha$-approximation
algorithm $A$
for the influence maximization problem on a single
network as input, and has approximation ratio
$\frac{(1 - \epsilon) \alpha}{(o + 1)k}$ for
arbitrary $\epsilon > 0$, $o$ is the number of
overlapping users, and $k$ is the number of layers
in the multiplex.
Experiments on real and synthesized multiplexes 
validate the efficacy of the proposed algorithms for
the problem of
influence maximization in the heterogeneous multiplex.
Implementations of ISF and KSN are available at
\url{http://www.alankuhnle.com/papers/mim/mim.html}.
\end{abstract}

\section{Introduction}\label{label:intro}
The rapid growth of large Online Social Networks (OSN) such as 
Facebook, Google+, and Twitter has enabled them to become thriving places for 
viral marketing in recent years. 
People are increasingly engaged in OSNs:
62\% of adults worldwide use social media and spend 22\% of 
online time on social networks on 
average \cite{osnstatistics}. Much like real-world social networks, 
information spreading in OSNs 
has viral properties, creating an excellent medium for marketing. 
Due to the impact of this effect on the popularity of new products, 
OSNs have rapidly become an attractive channel for raising awareness 
of new products or brands. In this context, an important 
problem is how to find the best set of seed users who can influence 
the most other users.

Increasingly, users engage in more than one OSN; 
they connect their accounts across multiple networks, 
such that posts in one network 
are simultaneously posted in other networks.
In Fig. \ref{fig:crosspost}, we show 
the process of connecting a Facebook and Twitter account, 
which allows automatically posting on Facebook when a new tweet is sent, and vice versa. As a consequence, 
the propagation of information can
cross from one OSN to another through these
overlapping users.

The influence propagation in each OSN will be particular to that network; 
for example, usage patterns for Twitter and Facebook are quite different.  
Moreover, even different
cascades in the same social network may be
better explained by different models 
of influence propagation \cite{saito}.
Thus, overlapping users connect
OSNs together into a multiplex structure of OSNs, comprising
multiple OSNs linked together through overlapping users, where each OSN has
different local propagation.
In this work, we study the multiplex
influence maximization (MIM) problem, to pick
the most influential seed nodes, 
on a multiplex of OSNs as described above.
Several natural questions arise: 
(1) what conditions on the propagation in each layer OSN are sufficient for the overall multiplex 
propagation to have the submodularity property, which
is important for approximation algorithms?
(2) Can existing methods for single OSNs be utilized within 
a solution to the
MIM problem? (3) What role do overlapping users play in the influence propagation on a multiplex of OSNs?
From a computational perspective, a multiplex consisting of multiple
OSN layers may be very large, comprising billions of users in each
layer. 

\begin{figure}%
\subfigure[Auto-post from Twitter to Facebook]{
		\centering
		\includegraphics[width=0.45\columnwidth]{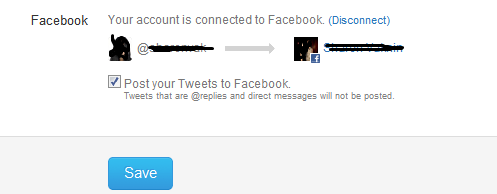}
		\label{fig:twitter2facebook}
	}	
	\subfigure[Auto-post from Facebook to Twitter] { 
		\centering
		\includegraphics[width=0.45\columnwidth]{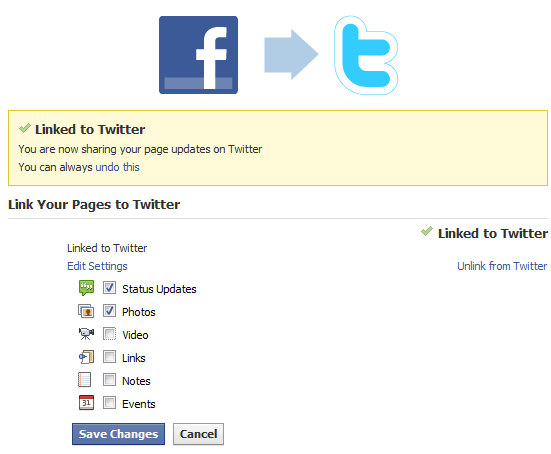}
		\label{fig:facebook2twitter}
	}	
\caption{The process of becoming an overlapping user
  of Twitter and Facebook.}%
\label{fig:crosspost}%
\end{figure}
\begin{figure}
  \centering
  \begin{tikzpicture}
    \begin{scope}[every node/.style={circle,thick,draw}]
      \node (0) at (0,0) {$a$};
      \node (1) at (2,0) {$b$};
      \node (2) at (2,2) {$c$};
      \node (3) at (0,-2) {$a$};
      \node (4) at (2, -2) {$b$};
\end{scope}

\begin{scope}[every node/.style={fill=white,circle},
  every edge/.style={draw=black,very thick}]
  \path [->] (0) edge node { $0.5$ } (2);
  \path [->] (1) edge node { $0.5$ } (2);
  \path [->] (3) edge node { $0.5$ } (4);
  \node (5) at (4,1) {Layer 1};
  \node (5) at (4,-2) {Layer 2};
\end{scope}

\begin{scope}[every node/.style={fill=white,circle},
  every edge/.style={draw=black,dashed}]
  \path [-] (0) edge (3);
  \path [-] (1) edge (4);
\end{scope}
\end{tikzpicture}
  \caption{A 2-layer multiplex exemplifying how propagation
    in a multiplex will deviate from the propagation in 
    a single layer.} \label{fig:example}
\end{figure}
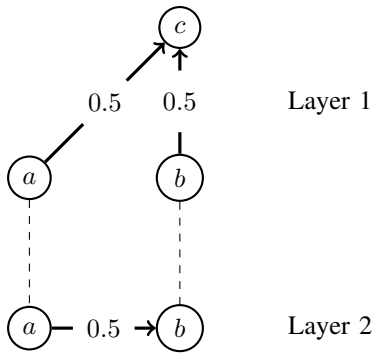 
To demonstrate how propagation in a multiplex differs
from propagation in a single layer, 
consider the toy multiplex shown in Fig. \ref{fig:example},
where $a,b$ are overlapping users in both layers, and
$c$ is only present in Layer 1.
Let Layer 1 have a fixed threshold model, with the
threshold $\theta_c$ of $c$ equal to 1. Thus, $c$ 
becomes activated iff $a,b$ are both activated. Let Layer 2
have the Independent Cascade (IC) model. 
Then seeding vertex $a$ will result
in both $b,c$ having a chance of becoming activated,
as $a$ may activate $b$ according to the IC model in
Layer 2. If $b$ becomes activated, then $a,b$ will together
activate $c$ in Layer 1. Finally, observe that the activation
of the two layers
cannot be incorporated into a single layer network following
either the IC model or the fixed threshold model.

For influence maximization problem on a single layer network with
propagation according to
a single model, such as Independent Cascade (IC),
approximation algorithms have been developed and optimized
\cite{Kempe2005, Chen2010, cohen2014sketch, Borgs2014max, Tang2014}.
However, since these algorithms only consider a single model of
influence propogation,
they are not directly
suitable for MIM, where each OSN has a different model of
propagation. As shown in Fig. \ref{fig:overlappingsizetable}, 
the fraction of overlapping users is considerable.

\begin{figure}%
\centering
\begin{tabular}{|c|c|c|c|c|} \hline
  - & LinkedIn & Facebook & Twitter & MySpace \\ \hline
  LinkedIn & - & $12 \%$ & $21 \%$ & $6 \%$   \\ \hline
  Facebook & $82 \%$ & - & $91\%$ & $57 \%$   \\ \hline
  Twitter & $31 \%$ & $20\%$ & - & $17 \%$   \\ \hline
  MySpace & $36 \%$ & $49\%$ & $70\%$ & -   \\ \hline
\end{tabular}
\caption{The percentage of overlapping users between major OSNs in 2009 \cite{Anderson2009}. The table is read as
follows: $(x,y)\%$ of users in OSN $y$ also use OSN $x$,
where $x$ is the row, $y$ is the column}%
\label{fig:overlappingsizetable}%
\vspace{-0.15in}
\end{figure}

Our main contributions are summarized as follows.
\begin{itemize}
	\item We define the
          generalized deterministic submodular (GDS) property,
          which if satisfied by each layer implies overall propagation 
          on the multiplex is submodular.  Submodularity allows 
          the formulation of a greedy $(1 - 1/e)$-approximation algorithm (ISF) for MIM.
	\item We provide an
          approximation algorithm
          Knapsack Seeding of Networks (KSN) which is 
          parallelizable by layer of the
          multiplex and utilizes
          an algorithm $\mathcal A$ for the single-layer influence maximization problem
          together with a knapsack-based approach
          to create a solution for MIM, thereby 
          taking advantage of previous optimizations \cite{Chen2010, cohen2014sketch, Borgs2014max, Tang2014} 
          for the homogeneous,
          single layer case.
          If the utilized algorithm $\mathcal A$ has approximation
          factor $\alpha$, then KSN has approximation ratio
          \[ \frac{(1 - \epsilon) \alpha}{(o + 1)k}, \]
          where $o$ is the number of overlapping users in 
          the multiplex,
          $k$ is the number of layers in the multiplex,
          and $\epsilon > 0$ is arbitrary.
         \item
          Experimental evaluation of all algorithms on
          a variety of multiplexes, both synthesized and from traces
          of real multiplexes of OSNs, 
          validates the effectiveness of the two approximation
          algorithms.

\end{itemize}

Overlapping users can be identified in real networks using,
for example, methods in \cite{iofciu2011identifying, buccafurri2012discovering}; 
methods of identification is not a focus of this paper.

The rest of the paper is organized as follows. In Section \ref{label:model}, 
we provide technical definition for model of influence
propagation.
Also, we present the influence propagation model in multiplex networks in terms of its component layers 
and define the problem. 
We then outline our proposed algorithms for solving the MIM problem 
along with the inapproximability proof for a class of models in Section \ref{label:solution}. 
Section \ref{label:experiment} shows the experimental results on the performance of different algorithms. 
In Section \ref{label:relatedwork} we discuss related work on
influence propagation and finally, Section \ref{label:conclusion} concludes the paper.  

\section{Model Representation and Problem Formulation}\label{label:model}
\subsection{Influence Propagation Models}
Intuitively, the idea of a model
of influence propagation in a network is clear; it is
a way by which nodes can be activated or influenced
given a set of seed nodes.  Kempe et al. studied
a variety of models in their seminal work on influence
progation on a graph, including
the Independent Cascade (IC) and 
Linear Threshold (LT) models \cite{Kempe2003}.
For completeness, we briefly describe these two models. 
An instance 
of influence propagation on a graph $G$ follows the IC model if a weight can be assigned
to each edge such that the propagation probabilities can be computed as
follows: once a node $u$ first becomes active, 
it is given a single chance to activate each currently inactive neighbor $v$ with probability proportional to 
the weight of the edge $(u,v)$. In the LT model each network user $u$ has an associated threshold $\theta(u)$ chosen uniformly from $[0,1]$ which determines how much influence (the sum of the weights of incoming edges) is required to activate $u$. $u$ becomes active if the total influence from its direct neighbors exceeds the threshold $\theta(u)$.

In this work, since we allow each layer of a multiplex
to have a different
model of influence propagation, we need
a technical definition for this concept.

\begin{mydef}[Model of influence propagation]
  A model of influence propagation $\sigma$ on 
  a graph $G = (V, E)$ is a function $P$ that assigns,
  for each $A \subset V$, and for each $S \subset V$,
  a probability 
  \[ P( S | A ) = P( S \text{ is final activated set } | A \text{ is seed set } ) \in [0, 1], \]
  satisfying
  \begin{align*}
    &(1) \text{ if } B \cap A \subsetneq A, 
    P(B | A) = 0, \\
    &(2) \sum_{S : S \subset V} P( S | A ) = 1.
  \end{align*}
  (1) simply states that seed nodes may not
      become unactivated, and (2) ensures that
      we have a probability distribution.

  The expected number of activated nodes given a seed set $A$
  is denoted $\sigma(A)$, and
  \[ \sigma (A) = \sum_{S : S \subset V} P( S | A ) \cdot |S|. \]
  A model $\sigma$ is called \emph{deterministic} iff
  for each $A \subset V$, there exists $F_A$ such that
  $P( F_A | A ) = 1$; intuitively, deterministic means
  that there is no probability in the model of diffusion,
  since the final set activated is uniquely determined
  by the seed set.  If $\sigma$ is deterministic,
  $\sigma(A) = |F_A|$; we abuse notation and also
  use $\sigma(A) = F_A$, the final set itself.  This allows
  convenient specification of the set $T = \sigma( \tau ( A ) )$, 
  for example, where both $\sigma, \tau$ are deterministic models,
  and $T$ is the final activated set generated by using the final
  set of $A$ under $\tau$ as the seed set for $\sigma$.
\end{mydef}



Many models of information propagation
discussed in the literature satisfy the
submodularity property, that $\sigma$ satisfies
\[\sigma (A) + \sigma (B) \ge \sigma (A \cup B) + \sigma( A \cap B ), \]  
for all $A, B \subset V$. Submodularity is important
since it guarantees that a greedy approach to the influence
maximization prolem will have an approximation ratio 
\cite{nemhauser78}. We now define a property
that is stronger than submodularity.

\begin{mydef}[Generalized Deterministic Submodular]
  Let $\sigma$ be a model of influence propagation.
  $\sigma$ satisfies the 
  \emph{generalized deterministic submodular property} 
  (GDS) if
  the expected number of activations, 
  given seed set $A$, can be written
  \[ \sigma (A) = \sum_{j = 1}^s p_j \sigma_{j} (A), \]
  where each $\sigma_{j}$, $j \in \{1, \ldots, s \}$ 
  is a deterministic,
  submodular model of influence propagation, and 
  $p_j \in [0, 1 ]$, $\sum_{j = 1}^s p_j = 1$.
\label{property1}
\end{mydef}

\begin{lemma} Let $\sigma$ be a model of
  influence propagation.  If $\sigma$ satisfies
  GDS, then $\sigma$ is submodular.
  \label{lemma01}
\end{lemma}
\begin{proof}
  Let $A$ be an arbitrary seed set.
  Since $\sigma$ satisfies
  GDS, 
  $\sigma(A) = \sum_{i = 1}^s p_j \sigma_j(A)$, where
  $\sigma_j(A)$ is expected activation of 
  deterministic and submodular model $\sigma_j$.
  Hence the expected activation function $\sigma$ 
  is a nonnegative linear combination of submodular
  functions, thus $\sigma$ is submodular.
\end{proof}

Examples of models in the literature that satisfy
submodularity include IC, LT, Asynchronous
Independent Cascade, Asynchronous Linear Threshold 
\cite{saito2010selecting},
Independent Cascade Model 
for Endogenous Competition,    
Homogeneous Competitive Independent Cascade model,
and K-LT competitive diffusion model, 
as well as others \cite{Kempe2003} \cite{chenwei13}.
In all cases,
the submodularity of these models has been shown
by considering instances where some edges
are live and some are blocked -- each such
instance corresponds to a deterministic
submodular model.  Thus, all of these proofs
show submodularity by showing the stronger property GDS
and using Lemma \ref{lemma01}. Another example of 
a model that satisfies GDS is the \textit{conformity and
context-aware cascade model} \cite{li2015conformity}.

In section \ref{label:solution}, we show (Theorem \ref{theorem1}) that influence propagation on a multiplex satisfies
GDS if the propagation on each layer network
satisfies GDS; hence if the propagation on each layer satisfies GDS, 
the propagation in the multiplex is submodular.
\subsection{Notations and Multiplex Model}
A social network can be modeled as a directed
graph $G = (V,E)$.
The vertex set $V$ 
represents the participation of users in the social network,
and the edge set
represents the connections among network users. 
These connections model friendships or 
relationships. 
\begin{mydef}[Heterogeneous multiplex]
A multiplex of OSNs is a list 
$\mathscr{G} = \{ (G_i, \sigma_i) : i \in \{1, \ldots, k \} \}$ with
$G_i = (V_i, E_i)$ a directed graph representing an OSN
and influence model
$\sigma_i$, representing the model of influence propagation
in $G_i$.
If a user belongs to
more than one OSN, an interlayer
edge is added between the pair of
nodes, one in each OSN, representing
this user. Such a user is termed
\emph{overlapping user}; we will denote the set of overlapping
users by $O$. We will denote the set of all users in the multiplex by  $V = \bigcup_{i = 1}^k V_i$. 

The influence propagation model $\sigma$ on the
multiplex is defined in the following way.
If an overlapping vertex $v$ is activated in one graph $G_i$, 
then, deterministically its adjacent interlayer copies
become activated in all OSNs;
propagation occurs in each graph 
$G_i$ according
to its propagation model $\sigma_i$.  
Figure \ref{fig:tf} 
shows an example of the definition of $\sigma$,
in a multiplex $\mathscr{G} = \{ (G_1, \sigma_1 ), (G_2, \sigma_2) \}$ with two layers.  Here $\sigma_1$ and $\sigma_2$ are simply deterministic models of activation following
the directed edges in the layers.
Initially, the activated set is $\{ v_1, v_6 \}$
in Fig. \ref{fig:tf}(a), shown by red nodes. In
\ref{fig:tf}(b), the activation has propagated according
to $\sigma_1, \sigma_2$ in $G_1, G_2$, respectively, 
activating in addition set $\{ v_2, v_3, v_8 \}$.
Next, in \ref{fig:tf}(c), the propagation proceeds
between the two layers via overlapping nodes, whereupon
it may continue in each layer $G_i$ according to $\sigma_i$.  Propagation ceases when no new nodes are activated
in any layer. In addition, Fig. \ref{fig:tf} demonstrates
how without loss of generality, we may consider all
nodes to be overlapping by adding absent nodes 
as isolated nodes (the white nodes).  In section
\ref{isf:submod}, we make use of this fact by
considering all layers to have the same nodes --
in the rest of the paper, we consider overlapping
users to be non-trivial; i.e. there are no isolated
nodes in any layer. We refer to the users that 
actively participate
in multiple networks (i.e. are non-isolated) 
as overlapping users, 
which may be identified in real networks using, for example,
methods in 
\cite{iofciu2011identifying, buccafurri2012discovering}. 

The expected number of activations in the multiplex
given seed set $A \subset V$ is denoted 
$\sigma (A)$, in addition to denoting the model
defined above as $\sigma$; we do not count more than one
copy of overlapping nodes towards $\sigma(A)$. 
Each graph $G_i$ is referred to
as a \emph{layer network} of the multiplex $G$. We refer to a graph $G = (V,E)$
that is not part of a multiplex as a single network, to contrast with multiplex network, and we refer to 
propagation occuring in a single network according to
a single influence propagation model as \emph{homogeneous}
 propagation, to contrast with the \emph{heterogeneous}
propagation in a multiplex with more than one layer and
propagation model in each layer.
\end{mydef}

\subsection{Problem Definition}
We now consider the problem of maximizing the influence of a seed set of given size in 
a multiplex network. Formally,

\begin{mydef}[Heterogeneous Multiplex Influence Maximization (MIM)]
Given a multiplex network $\mathscr{G} = (G_1, \sigma_1), \ldots, ( G_k, \sigma_k )$ with 
$k$ layer networks, influence
propagation model $\sigma$ on $\mathscr{G}$, as defined above in terms of $\sigma_i$,
and positive integer $l$, find a set $S \subset V$ of size at most $l$ 
so as to maximize the expected number of active users $\sigma (S)$.
An instance of this problem will be denoted $(\mathscr{G}, k, l, \sigma )$.
\end{mydef}

\section{Approximations of MIM}\label{label:solution}
Since influence maximization on a single network is a special case
of influence maximization on a multiplex, MIM is $NP$-complete.
In this section, we first prove that the propagation $\sigma$ on a multiplex is submodular
if the propagation on each layer satisfies GDS and formulate a greedy algorithm
to maximize expected influence.  If each layer satisfies GDS, $\sigma$ is
submodular and thus our greedy algorithm has approximation ratio $1 - 1/e$. 
Finally, we consider approaches to approximate MIM by
approximating influence maximization on each layer separately and 
effectively combining the result into a feasible solution for MIM,
which leads to a scalable approximation algorithm KSN with ratio depending on
number of overlapping users $o$, number of layers $k$ in the multiplex,
the ratio $\alpha$ for the approximation used on the homogeneous layers,
and an arbitrary $\epsilon > 0$; the ratio of KSN is 
$\frac{(1 - \epsilon) \alpha}{(o + 1)k}$.

\subsection{Greedy approach}
Let $\mathscr{G} = (G_i, \sigma_i)_{i = 1}^k$ be a multiplex with propagation model $\sigma$.  
We prove that $\sigma$ is submodular for the case that
each $\sigma_i$ satisfies generalized deterministic submodularity (GDS). 
Thus, the greedy algorithm, which we detail in this section, achieves a $(1-1/e)$ ratio
when propagation in each layer $\sigma_i$ satisfies GDS.

\subsubsection{Submodularity}
\label{isf:submod}

Without loss of generality, we may consider that 
the sets $V_i$ are the same;
that is, $V_i = V$ for all $i$ and some set $V$: 
if a vertex $v \in G_i$ does not exist in some 
$G_j$, simply add it 
to $G_j$ as an isolated vertex.  Thus, in this section
only, we consider $V_i = V$ 
for all $i$. Recall that instead of counting 
the activation of all $k$
copies of node $u \in V$, we count only a single copy as activated.
The expected number of activations in the multiplex
given seed set $A \subset V$ is denoted $\sigma (A)$; again,
we do not count more than one copy of $u \in V$ towards $\sigma(A)$. 

We will first consider a simpler case: when the propagation of each $G_i$ is deterministic and submodular.
\paragraph{Deterministic case}
In this section, let the propagation $\sigma_i$ of each $G_i$ 
be deterministic and submodular.  Recall the definition
of the multiplex influence propagation $\sigma$.
Given a seed set $S \subset V$, the set of nodes in the
multiplex that are
activated after the propagation finishes will be denoted by $\tau(S)$; the nodes activated if propagation is restricted to only $G_i$ will be denoted $\tau_i (S)$.  
Notice that $|\tau(S)| = \sigma(S)$ and $|\tau_i(S)| = \sigma_i(S)$.
\begin{figure}
  \centering 
  \subfigure[] {
    \includegraphics[width = 0.2\textwidth ]{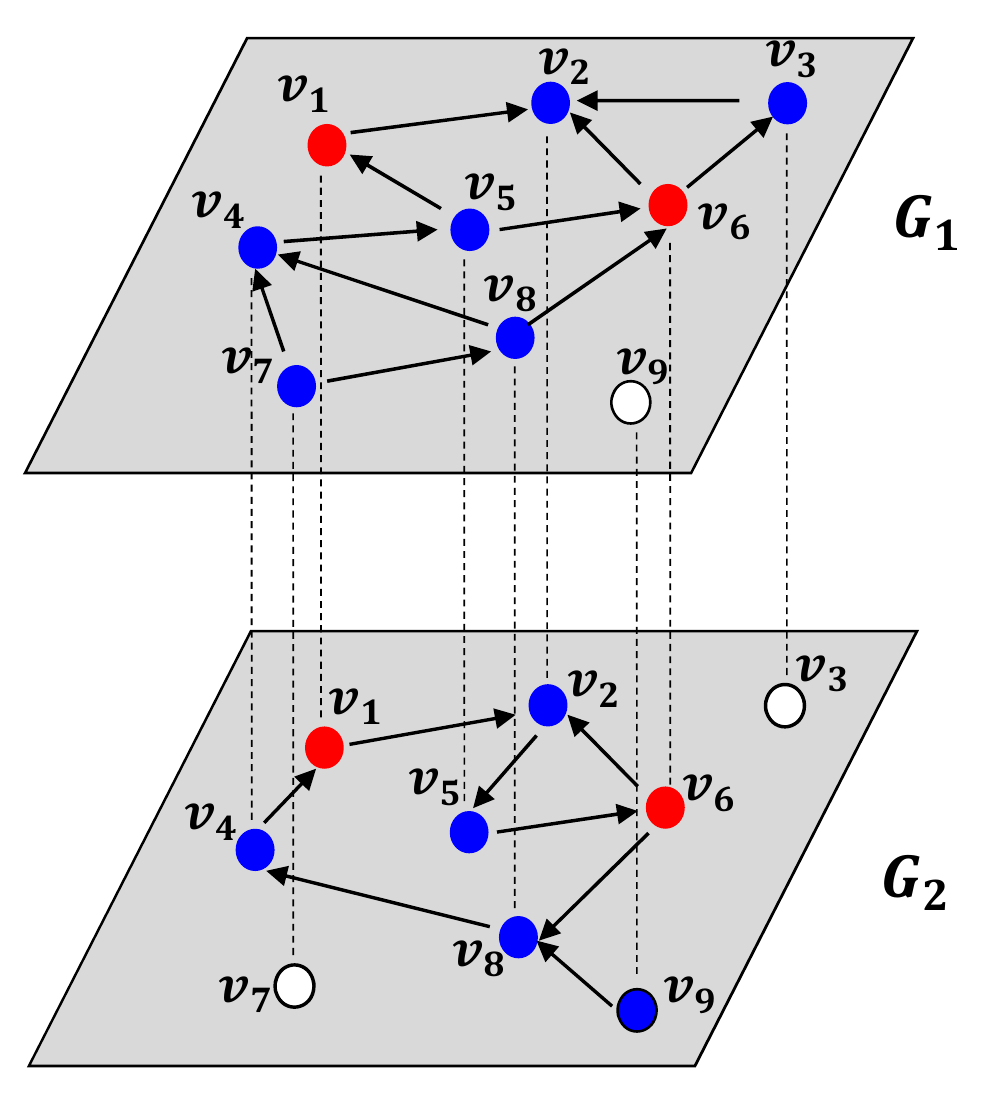}
		\label{fig:mult_frame}
	}
  \subfigure[] {
    \includegraphics[width = 0.2\textwidth ]{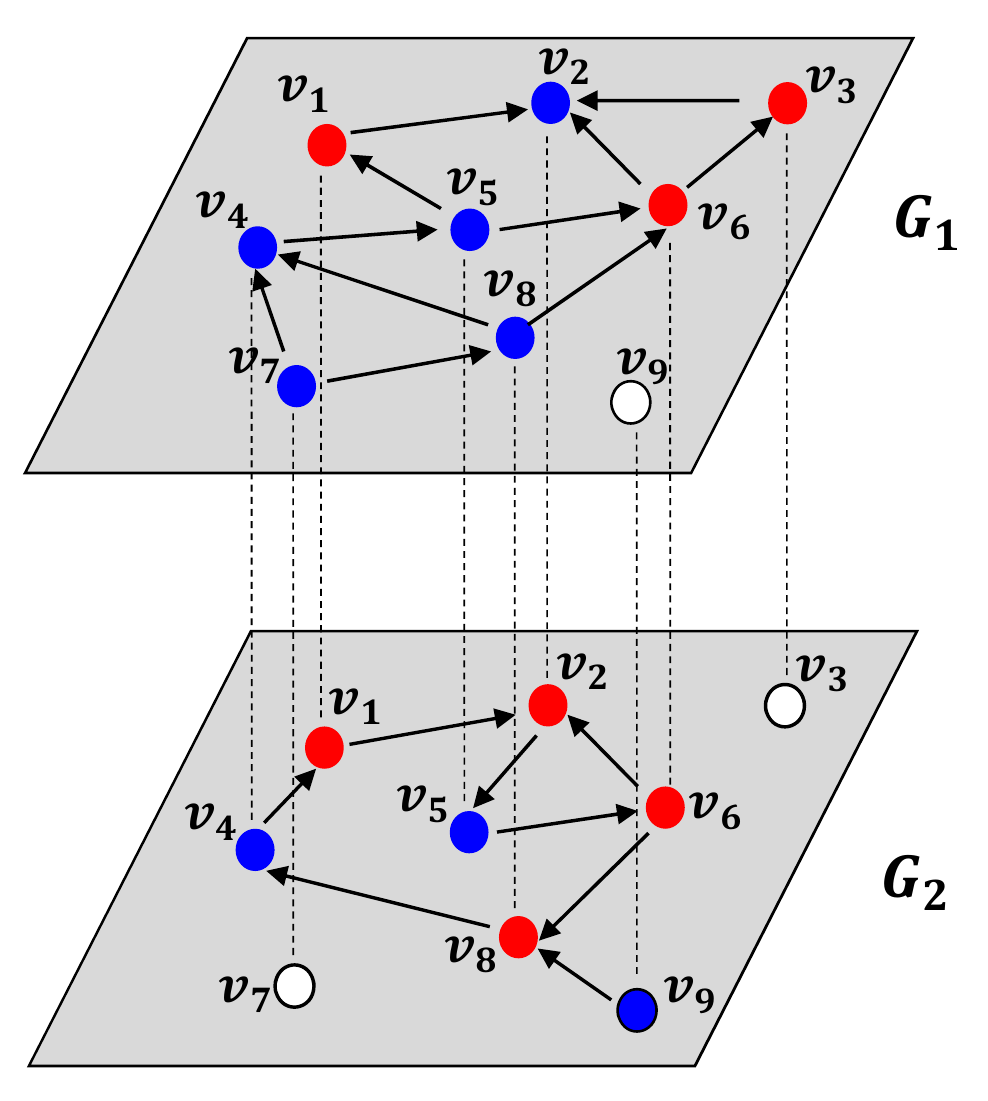} 
    \label{fig:mult_fig2}
  }

  \subfigure[] {
    \includegraphics[width = 0.2\textwidth ]{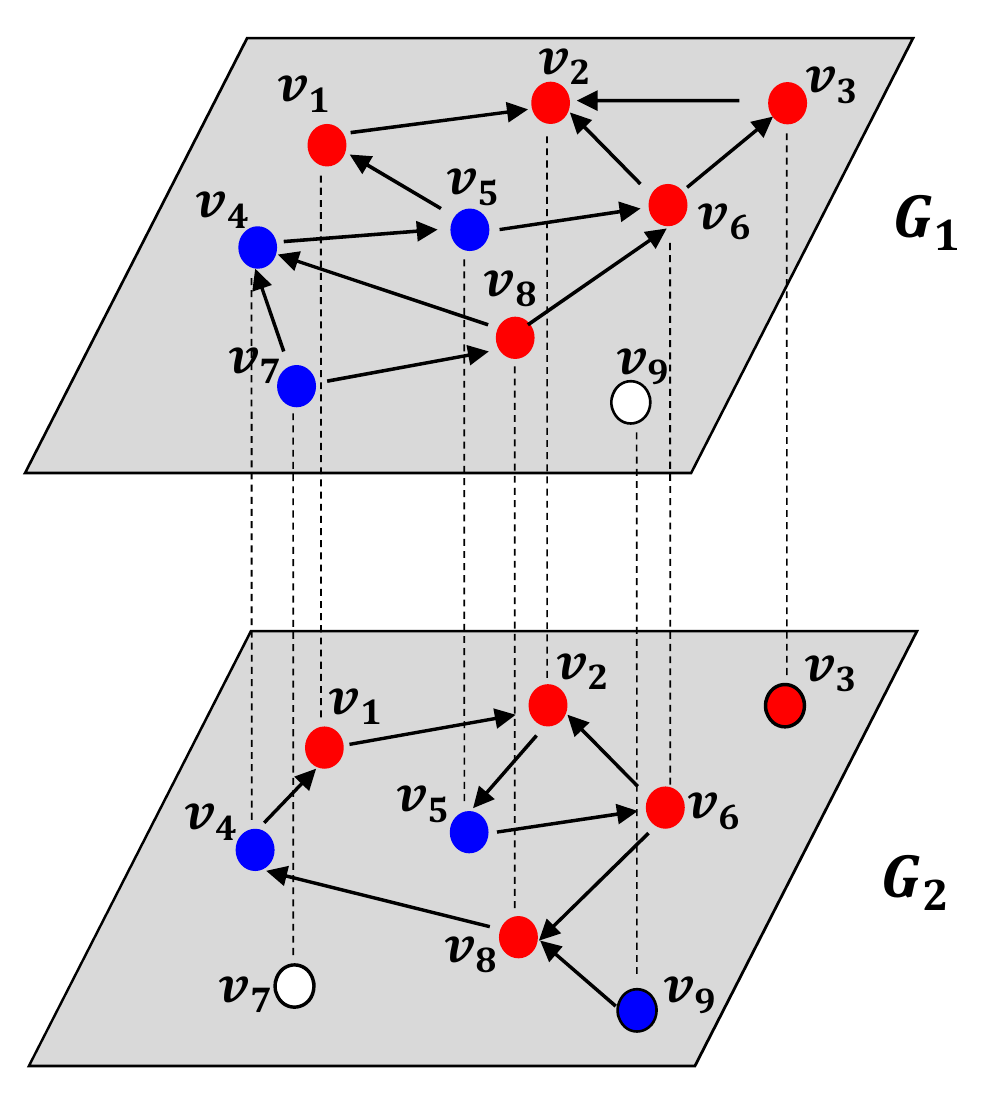} 
    \label{fig:mult_fig3}
  }

  \caption{(a)--(c) An example of influence 
    propagation $\sigma$ in the 
    multiplex structure.  Red nodes are activated, blue and
    white are unactivated nodes.  The white color indicates
    a node added to the layer in order that each network
    has the same vertex set $V$. In (a), we have seed set
    conserting of $v_1$ and $v_6$.  In (b), propagation
  occurs in each layer according to that layers model, and
  in (c), nodes activated in one layer are activated in all,
  whence the propagation may continue.}

  \label{fig:tf}
  \vspace{-16pt}
\end{figure}
\begin{lemma}
  Let $S \subset V$.  Then 
  $\tau_i ( \tau( S ) ) = \tau( S )$ for all $i$. 
  \label{lemma1}
\end{lemma}
\begin{proof}
This follows from definition of propagation in multiplex. If propagation has resulted
in set $\tau( S )$, then propagation cannot proceed further in any layer network,
for otherwise propagation in $\mathscr{G}$ would not have terminated
with $\tau ( S)$. A visualization of an example
of multiplex propagation is shown in Figure \ref{fig:tf} (a)--(c).
\end{proof}
\begin{lemma}
  Let $S, T \subset V$.
    \[ \tau(S) \cup \tau(T) = \tau_i ( \tau(S) \cup \tau(T) ) \]
      for all $i$.
      \label{lemmax}
\end{lemma}
\begin{proof}
  We have
      \[
        \sigma_i (\tau (S) \cup \tau (T)) + \sigma_i (\tau (S) \cap \tau (T)) \le \sigma_i (\tau (S)) + 
    \sigma_i (\tau (T)) \]
  by submodularity of $\sigma_i$, so
  \begin{align*}
    \sigma_i (\tau (S) \cup \tau (T)) &\le |\tau (S)| + |\tau (T)| - |\tau (S) \cap \tau (T)|  \\
    &= |\tau (S) \cup \tau (T)|, 
  \end{align*}
  by Lemma \ref{lemma1} and since $ \sigma_i (\tau (S) \cap \tau (T))  \ge |\tau (S) \cap \tau (T)|$.  Hence,
  \[ \tau_i (\tau (S) \cup \tau (T)) = \tau (S) \cup \tau (T). \]
\end{proof}
\begin{lemma}
  Let $S, T \subset V$.
  \begin{itemize}
    \item[i)]
      \[\tau ( S \cup T ) = \tau(S) \cup \tau(T) \]
    \item[ii)]
      \[\tau (S \cap T) \subset \tau_i ( \tau (S) \cap \tau (T)) .\]
  \end{itemize}
  \label{lemma0}
\end{lemma}
\begin{proof}
  \begin{itemize}
    \item[i)]

  Since $S \cup T \subset \tau (S) \cup \tau (T)$, $\tau (S) \cup \tau (T) \subset \tau( S \cup T)$
  and propagation in any layer network $G_i$ cannot procede beyond $\tau (S) \cup \tau (T)$ by Lemma 
  \ref{lemmax},
  we have $\tau (S \cup T) = \tau (S) \cup \tau (T)$.
  \item[ii)]
    Clearly $\tau (S \cap T) \subset \tau (S)$ and similarly
    $\tau (S \cap T) \subset \tau (T)$, hence
    \[ \tau (S \cap T) \subset \tau (S) \cap \tau (T) \subset \tau_i (\tau (S) \cap \tau (T)), \]
    for any $i$.
  \end{itemize}
\end{proof}

\begin{lemma}
  If the propagation model $\sigma_i$ on each component of $\mathscr{G}$ is deterministic and submodular, the (deterministic) propagation $\sigma$ in 
multiplex $\mathscr{G}$
will be submodular.
\label{det_lemma}
\end{lemma}
\begin{proof}
  Let $i \in \{1, \ldots, n \}$. Then by Lemmas \ref{lemma1}, \ref{lemmax} and \ref{lemma0} and submodularity
  of $\sigma_i$,
  \begin{align*}
    \sigma (S \cap T) + \sigma (S \cup T) & = |\tau( S \cap T )| + |\tau( S \cup T )| \\
    & \le |\tau_i ( \tau (S) \cap \tau (T) )| + |\tau_i ( \tau (S) \cup \tau (T) ) | \\
    &= \sigma_i (\tau (S) \cap \tau (T)) + \sigma_i (\tau (S) \cup \tau (T) ) \\
    & \le \sigma_i ( \tau (S)) + \sigma_i (\tau (T)) \\
    &= \sigma(S) + \sigma(T).
  \end{align*}

  Thus, $\sigma$ is submodular.
\end{proof}
\paragraph{Probabilistic case}
Now the result for the deterministic case is generalized to the case when all networks satisfy GDS.
\begin{theorem}
  Given multiplex network $\mathscr{G}$ with $k$ layer networks $G_i$, 
  if the model $\sigma_i$ on each layer network $G_i$ satisfies GDS
  then $\sigma$
  satisfies GDS.
\label{theorem1}
\end{theorem}
\begin{proof}
Since $\sigma_i$ satisfies GDS, with probability $p_{ij}$, $G_i$ has deterministic, submodular propagation $\sigma_{ij}$, such that $\sigma_i = \sum p_{ij}\sigma_{ij}$.
The probability that $\sigma$ will comprise $\sigma_{1j_1},  \sigma_{2j_2}, \ldots, \sigma_{kj_k}$ is
$\prod_{i = 1}^k p_{ij_i}$, since propagation in each graph is independent.  
Let this propagation in $\mathscr{G}$ be labeled $\sigma_{j_1,\ldots,j_k}$. By Lemma \ref{det_lemma},
$\sigma_{j_1,\ldots,j_k}$ is submodular and deterministic.
\end{proof}

\subsubsection{$(1-1/e)$-Approximation Algorithm}
In this section we detail the greedy algorithm \emph{Influential Seed Finder} (ISF) for solving the MIM problem. As shown in Alg. 1, ISF is a greedy algorithm (with CELF++ optimization \cite{goyal2011celf++}), which chooses a node that maximizes the marginal gain of $\sigma$ at each iteration. Recall that $\sigma$
is the expected activation of the influence model $\sigma$ defined on the multiplex
in Section \ref{label:model}, which incorporates the models $\sigma_i$ on each
layer utilizing the overlapping users.  To compute $\sigma$, it is necessary to
compute expected activation $\sigma_i$ on each layer. To perform this computation, we use
independent Monte Carlo simulations -- in general, $\sigma_i$ could be any model
of influence propagation, and thus may not be amenable to specialized techniques for triggering model \cite{Tang2014}.

As shown earlier, $\sigma$ is submodular and monotone increasing when each individual 
network satisfies 
GDS; therefore, in this case ISF has an 
approximation ratio of $(1 - 1/e)$ \cite{nemhauser78}.
The time complexity of ISF is $O(nl(m+n) \log n)$ where 
$n, m$ are number of users, friendships in the
multiplex of OSNs, respectively.  Each Monte Carlo
simulation takes time $\Omega ( n + m )$, and the $\log n$ factor accounts for the time to adjust the priority queue,
$l$ is the size of the seed set chosen.

\begin{algorithm}
  \caption{ISF: An algorithm for finding the best seed users. Approximation ratio: $1 - 1/e$ when
each layer satisfies GDS property}
\textbf{Input:} A multiplex $\mathscr{G} = (G^1, G^2, \ldots, G^k)$, $l$\\
  \textbf{Output:} Seed set $S$ of size $l$
  \begin{algorithmic}[1]
  \State Renumber all the nodes across all networks so that each node gets a unique id
  \State $S \leftarrow \emptyset$
  \State $V \leftarrow \cup_{i=1}^{k} V_i$
  \For{each $v \in V$}
	\State $v.marginal\_gain =   	\sigma(v)$ 
  	\State $v.round = 0$	
  \EndFor
  \State Initialize max priority queue $Q$ with (key,value) pair $(v,v.marginal\_gain)$, $\forall v \in V$
  \State Initialize previous marginal gain, $prev\_mg = 0$
  \While{$|S| \le l$}
  	\State $v \leftarrow Q.pop().key$
  	\If{$v.round == S.size$}
  		\State $S \leftarrow S \cup \{v\}$
  		\State $prev\_mg \leftarrow prev\_mg + v.marginal\_gain$
  	\Else
  		\State $v.marginal\_gain \leftarrow \sigma(S \cup \{v\}) - prev\_mg$
  		\State $v.round = S.size$
  		\State $Q.add(v,v.marginal\_gain)$	
  	\EndIf 
  \EndWhile
  \State \textbf{Return } $S$
	\end{algorithmic}
	\label{alg:greedyHMI}
\end{algorithm}

\subsection{Parallelizable multiplex algorithm}
Although in the case that the model of
propagation on each layer satisfies
GDS, we have the $(1 - 1/e)$ performance guarantee 
of the greedy ISF, the 
running time
of ISF may be impractical for large 
network sizes;
hence we propose Alg. \ref{alg:KSN} (KSN), another approximation algorithm which
parallelizes the problem in 
terms of the component layers -- the difficulty lies in combining
the solutions to the influence maximization problem on the separate layers
to obtain a solution for MIM.  KSN achieves this by approximating the solution to
multiple-choice knapsack problem.
The approximation ratio of KSN depends on the number of
overlapping users $o$, the number of layers $k$,
an arbitrary $\epsilon > 0$, and the ratio $\alpha$
of its input homogeneous layer algorithm $A$.

\subsubsection{Description of KSN} 
\begin{figure}%
\centering
\begin{tabular}{|c|c|c|c|c|c|} \hline
          - & 0 & 1 & 2 & 3 & 4 \\ \hline
  $G_1$ & 0 & 200 & 350 & 400  & 425   \\ \hline
  $G_2$ & 0 & 600 & 601 & 602 & 603  \\ \hline
  $G_3$ & 0 & 200 & 210 & 214 & 214   \\ \hline
\end{tabular}
\caption{An example of how KSN works, as described in the text.}
\label{fig:ksn-table}%
\end{figure}
The KSN algorithm takes as input an algorithm $A$ (with ratio $\alpha$) to solve the influence maximization problem on a single layer network, a multiplex network 
$\mathscr{G}$ with $k$ layers, and number of seeds to pick $l$.  For each $j \in \{1, \ldots, l \}$,
$i \in \{1, \ldots, k \}$, algorithm $A$ is run in parallel on each $G_i$ to 
get seed sets $T_{ij}$ with $j$ seed nodes.
It then uses an approximation to the multiple-choice knapsack problem (defined below) to decide 
how many nodes should be seeded in each layer, i.e. for each $i$, which $T_{ij}$ to pick.

For example, suppose we have a multiplex with three layers: $G_1,G_2,G_3$. Using 
algorithm $A$, we generate the table in Fig. \ref{fig:ksn-table}, where
the $(i,j)$th entry gives the activation of
seeding $j$ nodes in layer $G_i$. We then use an algorithm for multiple-choice knapsack to
choose for each layer $G_i$, the number $j_i$ of nodes to seed in that layer.

\begin{algorithm}
  \caption{Knapsack Seeding of Networks (KSN): A knapsack approach to finding the best seed users. Approximation ratio: $\frac{(1 - \epsilon) \alpha}{(o + 1)k}$, where $\epsilon > 0$, $o$ is the number of overlapping users, $k$ is the number of layers, and $\alpha$ is the ratio of algorithm $A$ on homogeneous networks}
  \textbf{Input:} Algorithm $A$, a multiplex network $\mathscr{G} = (G^1, G^2, \ldots, G^k)$, $l$\\
  \textbf{Output:} Seed set $T$ of size $l$
  \begin{algorithmic}[1]
  \For{$i \in \{1, \ldots k \}$}
      \State Run algorithm $A$
      on $G_i$ with input $j$
      to get seed sets $T_{i1}, T_{i2}, \ldots, T_{il} \subset G_i$, with $| T_{ij} | = j$.
  \EndFor
  \State For each $T_{ij}$, let cost $c(T_{ij}) = | T_{ij} |$, and profit $p(T_{ij}) = \sigma ( T_{ij} )$.
  \State Use $(1 - \epsilon)$-approximation to MCKP to
  choose for all $i$, $T'_{i} \in \{ T_{i1}, \ldots, T_{il} \}$, which choice satisfies $\sum_{i = 1}^k |T'_i| = l$.
  \State \textbf{Return } $T = \bigcup_i T'_i$.
  \end{algorithmic}
	\label{alg:KSN}
\end{algorithm}
\paragraph{Worst-case bound on performance of KSN}
\label{KSN-perf}
First, we need the definition of the multiple-choice
knapsack problem.
\begin{mydef}[Multiple-choice knapsack problem (MCKP)]
Let $(\mathscr{C}, k, l, c, p, B)$ be given, where
$\mathscr{C} = \{C_1, \ldots C_k \}$ 
comprises $k$ classes of $l$ objects, $C_i = \{ x_{ij} : 1 \le j \le l \}$, $c$ and $p$ are cost and profit
functions on objects $x_{ij}$, and budget $B \ge 0$.  
The multiple-choice
knapsack problem (MCKP) is to pick one item from each class, $x_i'$ such that 
profit $\sum_{i = 1}^k p(x_i')$ is maximized under the constraint $\sum_{i = 1}^k c(x_i') < B$.
\end{mydef}

For $\epsilon > 0$, MCKP has a $(1 - \epsilon)$-approximation as shown in \cite{mckp}. We will use this algorithm
to obtain an approximation for MIM as follows.
Let an instance $(\mathscr{G}, k, l)$ of MIM be given.  
For each pair $(i,j)$, $1 \le i \le k, 1 \le j \le l$,
let $T^{opt}_{ij}$ be an optimal seed set for 
$G_i$ satisfying the two conditions 
$T^{opt}_{ij} \subset G_i$ and $|T^{opt}_{ij}| = j$.
In addition, let $T_{ij}$ be the approximation from 
algorithm $A$
to $T^{opt}_{ij}$.  That is, $T_{ij} \subset G_i$, 
$|T_{ij}| = j$, and 
\[ \sigma( T_{ij}^{opt} ) \le \alpha^{-1} \sigma ( T_{ij} ). \]

Then let $C_i = \{ T_{i0}, \ldots, T_{il} \}$, 
$C^{opt}_i = \{ T^{opt}_{i0}, \ldots, T^{opt}_{il} \}$.
Finally, let $\mathscr{C} = \{ C_i : 1 \le i \le k \}$,
$\mathscr{C}^{opt} = \{ C_i^{opt} : 1 \le i \le k \}$,
and for each $i, j$, define $c ( T_{ij} ) = j$, $p ( T_{ij} ) = \sigma ( T_{ij} )$, and likewise define
$c, p^{opt}$ for each $T_{ij}^{opt}$.
Thus, we have two instances of the knapsack problem,
namely $I_1 = (\mathscr{C}, k, l, c, p, l)$ and  
$I_2 = (\mathscr{C}^{opt}, k, l, c, p^{opt}, l)$.

\begin{lemma}
  Let $Opt_{I_i}$ be the value of the optimal solution
  to MCKP instance $I_i$, $i \in \{1, 2 \}$. Then
  \[ \alpha Opt_{I_2} \le Opt_{I_1}. \]
  \label{Lemma_mckp}
\end{lemma}
\begin{proof}
  Suppose 
  $\{T_{ib_i} : 1 \le i \le k \}$, $\{ T^{opt}_{ia_i} : 1 \le i \le k \}$ are the optimal
  solutions for $I_1, I_2$, respectively.  Then 
  \begin{align*}
    \alpha Opt_{I_2} &= \alpha \sum_{i} \sigma( T^{opt}_{ia_i} ) \\
    &\le \sum_i \sigma (T_{ia_i}) \\
    &\le \sum_i \sigma (T_{ib_i} ) \\
    &= Opt_{I_1}, \\
  \end{align*}
  since algorithm $A$'s selection of $T_{ia_i}$ ensures
  $\sigma (T_{ia_i} ) \ge \alpha \sigma (T^{opt}_{ia_i})$. The last
  inequality follows from the fact that $\{ T_{ia_i} \}$ is a feasible
  solution to instance $I_1$, and $\{ T_{ib_i} \}$ is the optimal solution
  to $I_1$.
\end{proof}
\begin{theorem} Let $A$ be an $\alpha$-approximation to the
  problem of influence maximization on a homogeneous single 
  layer, and let $o, k$ be the number of overlapping users and layers, respectively in the
  multiplex.  Furthermore, suppose the propagation $\sigma_i$
  on each layer of the multiplex is submodular.
  Then, KSN has approximation ratio
  $\frac{ (1 - \epsilon)\alpha }{(o + 1)k }$.
\end{theorem}
\begin{proof}
  Suppose KSN returns the union of
  $T_{1a_1}, T_{1a_2}, \ldots, T_{1a_k}$, selected from
  $I_1$. Let $S_{opt}$ be the optimal solution to MIM instance $(\mathscr{G}, k, l)$. 
  Let $\sigma ( S_{opt} )^i$ denote the expected activation under $\sigma$ in layer
  $G_i$. Immediately, we have
  \begin{equation}
    \sigma ( S_{opt} ) \le \sum_{i = 1}^k \sigma (S_{opt})^i. \label{ineq1}
  \end{equation}
  Also, letting $O$ be the set of overlapping users, we have
  \begin{equation}
    \sigma ( S_{opt} )^i \le \sigma_i (S_{opt} \cup O ) \le \sigma_i(S_{opt}) + \sigma_i(O), \label{ineq2}
  \end{equation}
  where the first inequality in (\ref{ineq2}) follows from 
  the fact that any activation in $G_i$ proceeds according
  to the model $\sigma_i$ and results from seed nodes in $S_{opt} \cap G_i$ or through overlapping users
  $O$. The second inequality in (\ref{ineq2}) follows from submodularity
  of $\sigma_i$.

  Recall that $OPT_{I_j}$ denotes the optimal value of MCKP on instance
  $I_j$, for $j = 1,2$ as defined above; 
  let $KSN$ denote the value of the solution
  returned by Alg. KSN. Then,
  $\frac{1}{1 - \epsilon } KSN \ge OPT_{I_1}$; finally, notice if $S$ is any set of size
  at most $l$, and $i$ is a fixed layer,
  \begin{equation} 
    \frac{1}{(1 - \epsilon) \alpha }KSN \ge OPT_{I_2} \ge \sigma_i( S ), \label{ineq3}
  \end{equation}
  by Lemma \ref{Lemma_mckp}, and since $\sigma_i (S)$ is the value of a feasible solution to MCKP instance 
  $I_2$.
  Therefore, by (\ref{ineq1}), 
  (\ref{ineq2}), and (\ref{ineq3}), we have
  
  \begin{align*}
    \sigma ( S_{opt} ) &\le \sum_{i = 1}^k \sigma (S_{opt})^i \\
    &\le \sum_{i = 1}^k \sigma_i (S_{opt} ) + \sum_{i = 1}^k \sigma_i (O) \\
    &\le \frac{k}{(1 - \epsilon ) \alpha }  KSN + \sum_{i = 1}^k \sigma_i ( O ) \\
    &\le  \frac{k}{(1 - \epsilon ) \alpha } KSN + \sum_{v \in O} \sum_{i = 1}^k \sigma_i ( v ) \\
    &\le  \frac{k}{(1 - \epsilon ) \alpha } KSN +  \frac{ok}{(1 - \epsilon ) \alpha } KSN \\
    &\le \frac{(o + 1)k}{( 1 - \epsilon )\alpha} KSN.
  \end{align*}

\end{proof}

\paragraph{Time complexity of KSN}
\label{KSN-tc}
KSN runs algorithm $A$ in parallel 
$l \cdot k$ times, then employs
the $(1 - \epsilon)$ MCKP algorithm from \cite{mckp}.
Thus if $tc(A, G_i, j)$ is the time complexity of $A$ on $j$
seed nodes with graph $G_i$, 
the time complexity of KSN is
\[ O \left( \max_{(i,j) = (1,1)}^{(k,l)} \, tc(A, G_i, j) + (kl)^{\lceil 1/\epsilon - 1 \rceil} \log k  \right) , \]
since $O( (kl)^{\lceil 1/\epsilon - 1 \rceil} \log k )$ is the time complexity 
for the $(1 - \epsilon)$ MCKP algorithm with $k$ classes and $l$ items in each class.

Notice that the scalability of KSN depends on the scalability of the input
algorithm $A$.  For example, in the special case that each $\sigma_i$
satisfies the triggering model \cite{Kempe2003} and 
also satisfies each $\sigma_i$ submodular, then letting
algorithm $A$ be the TIM algorithm from \cite{Tang2014},
we would have the expected running time of KSN bounded by
\[ O((k + \ell)(m + n) \log n / \epsilon^2  + (kl)^{\lceil 1/\epsilon - 1 \rceil} \log k ), \]
where $n$ is the maximum number of nodes in a layer, $m$
is maximum number of edges in a layer, $\ell$ is an integer; and approximation ratio
\[ \frac{ (1 - \epsilon) (1 - 1/e - \epsilon ) }{ (o + 1)k } \]
with probability $(1 - n^{-\ell})^k$.

\section{Experimental Results}\label{label:experiment}
In this section, we perform experiments on both synthesized and real-world networks to show 
the effectiveness of the proposed algorithms. 
\subsection{Methodology}
We evaluated the following algorithms:
\begin{itemize}
  \item ISF (Alg. \ref{alg:greedyHMI}), the greedy algorithm with CELF++ optimization on the multiplex,
  \item KSN (Alg. \ref{alg:KSN}), with
    algorithm $A$ in the definition
    of Alg. \ref{alg:KSN} is set to the 
    CELF++ algorithm \cite{goyal2011celf++} or
    the IMM algorithm \cite{Tang2014},
  \item Even Seed (ES), which seeds 
    each layer of the multiplex with an equal number 
    of seed nodes $l / k$,
  \item Best Single Network (BSN), which
    places all $l$ seed
    nodes in the layer that maximizes
    $\sigma_i ( S_i )$, where $S_i$ is the seed set 
    chosen according
    to CELF++ in layer $i$, with $|S_i| = l$.
\end{itemize}
To estimate the expected activation $\sigma$
on the multiplex, or $\sigma_i$ in layer $G_i$, we use
independent Monte Carlo simulations.

Since the greedy CELF++ approach is not 
very scalable, we limit the maximum 
length of a diffusion sequence to 4 in Sections \ref{sect:synth}, \ref{sect:real}. The experiments 
in Sections \ref{sect:synth}, \ref{sect:real} were run on a machine with
an Intel(R) Xeon(R) W350 CPU and 12 GB RAM.

In Section \ref{sect:ksn-scale}, we demonstrate
the scalability of KSN on large multiplexes. 
This implementation\footnote{Source code available at \url{http://www.alankuhnle.com/papers/mim/mim.html}} of KSN 
is parallelized and utilizes algorithm $A$ set to 
the IMM algorithm of Tang et al. \cite{Tang2015}.
We chose IMM since it is highly scalable and source
code is available to solve the single layer problem
with both IC and LT models. For the MCKP problem within
KSN,
we used our implementation of
the 1/2-approximation from Chandra et al. \cite{mckp}.
These experiments were run on a machine with 
2 Intel(R) Xeon(R) CPU E5-2697 v4 @ 2.30GHz 
and 384 GB RAM.

\subsection{Synthesized multiplexes} \label{sect:synth}
We consider synthesized multiplexes based on
three scale-free networks $H_1$, $H_2$ and $H_3$ generated
according to Barabasi-Albert model \cite{barabasi1999emergence} with 1000 nodes and 4000 edges, with average degree 4; the exponent in the
power-law degree distribution
generated by this method is 2, which is consistent with that observed
in real-world social networks which have exponents 2 -- 3,
hence these synthesized networks 
should act as a good representative for capturing the social influence 
spread phenomenon. We assigned $H_1$, $H_2$ and $H_3$ with diffusion
models LT, IC and MLT respectively, with edge weights
and thresholds chosen uniformly in $[0,1]$,
where MLT is the deterministic, nonsubmodular majority linear threshold
model \cite{Chen2008}, whereby a node is activated if majority
of its neighbors are activated; then,
to form the multiplexes, beginning with a specified number of
overlapping users $o$,
we select the overlapping users 
randomly, such that each overlapping user is in
all three of the layers; that is, to create an overlapping node, we
randomly choose indices from each layer, and add two interlayer edges
to connect these three separate users into a single overlapping user.
This step is repeated until we have $o$ overlapping users in the multiplex.
A multiplex created in this way will be called a scale-free (SF)
multiplex, and will be denoted $\mathscr{H}_o$,
where $o$ is the number of overlapping users. 

\begin{figure*}
  \centering 
  \subfigure[Overlap = 0 (Scale-free, $\mathscr{H}_0$)] {
		\includegraphics[width=0.25\textwidth]{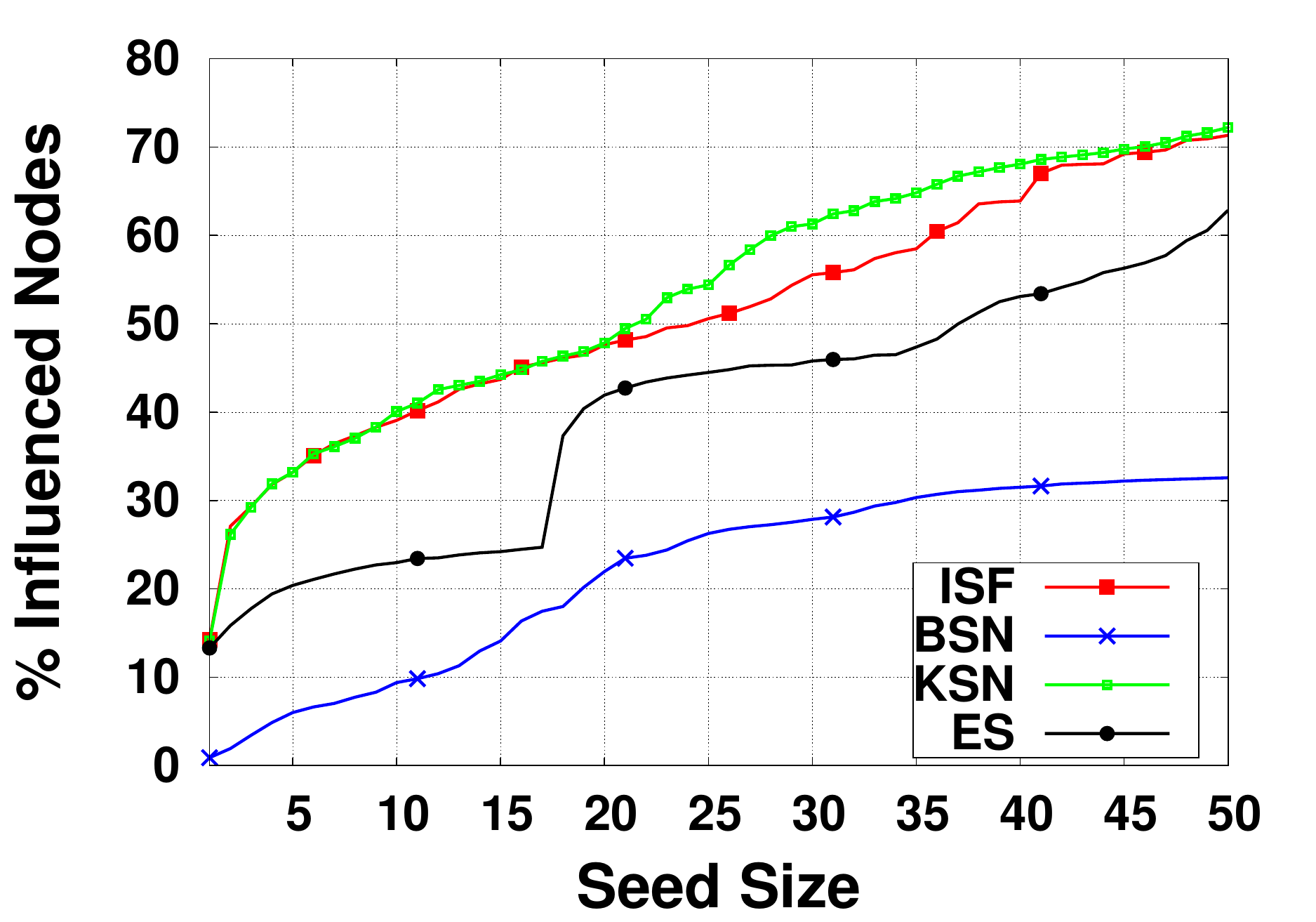}
		\label{fig:sf_syn_0ov}
	}
  \subfigure[Overlap = 500 (Scale-free, $\mathscr{H}_{500}$)] {
		\includegraphics[width=0.25\textwidth]{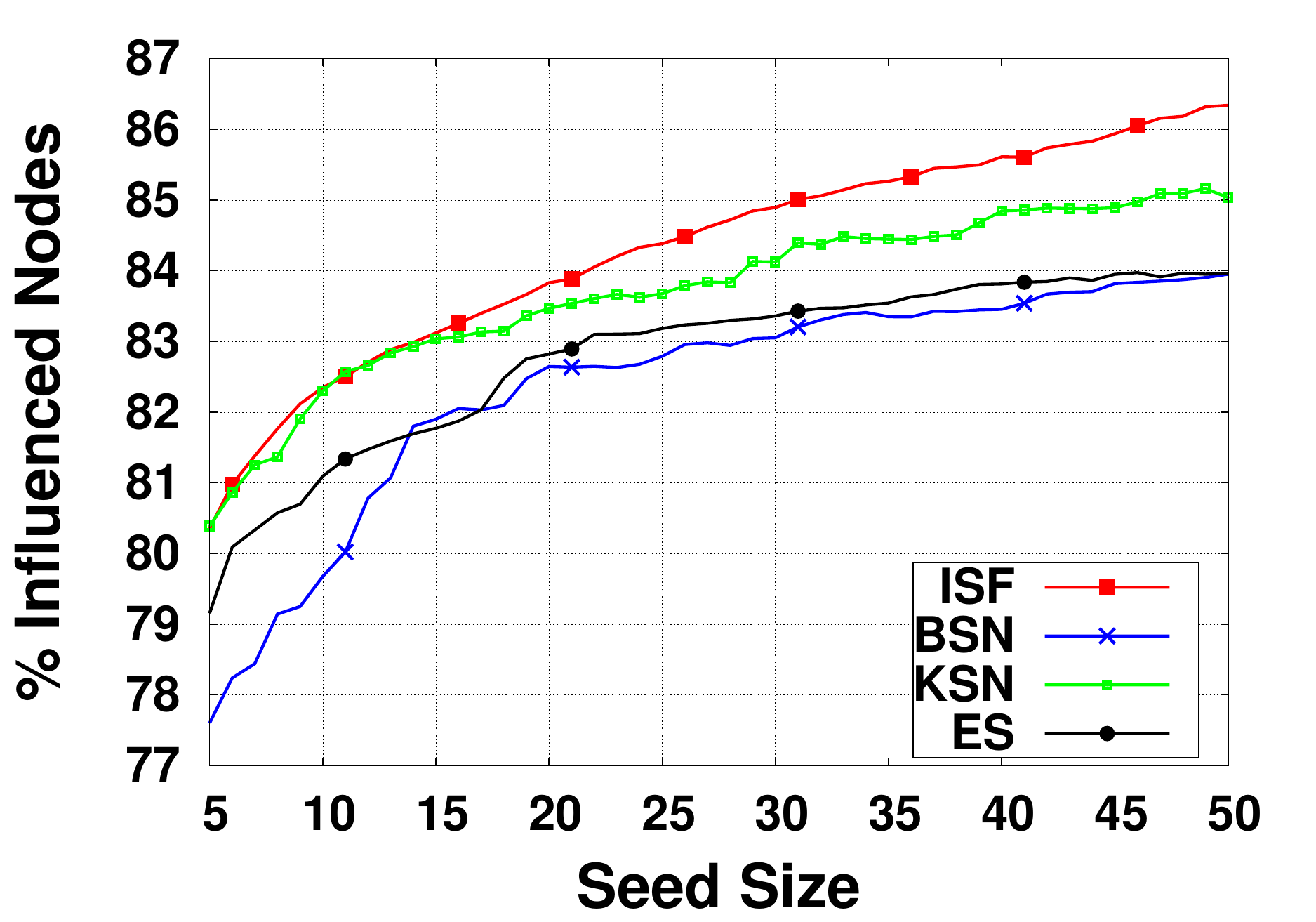}
		\label{fig:sf_syn_500ov}
	}

  \subfigure[Scale-free, $ \mathscr{H} $] {
		\includegraphics[width=0.25\textwidth]{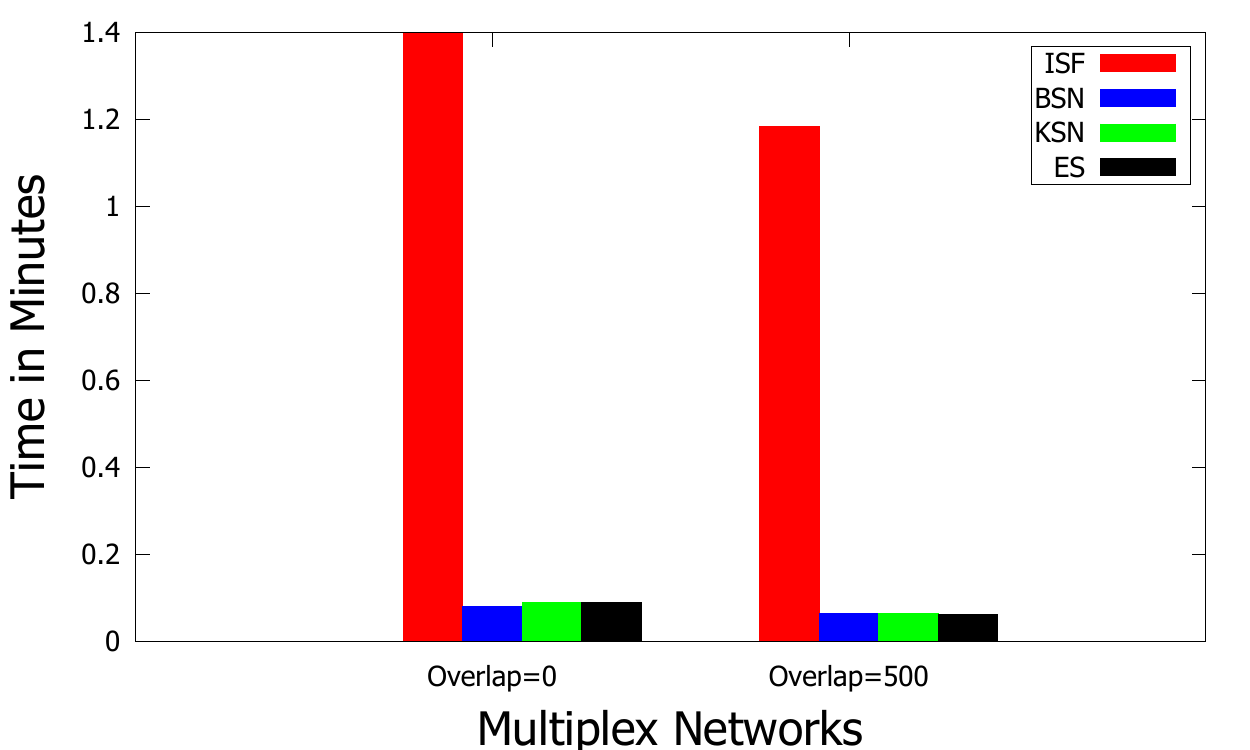}
		\label{sf_time}
    }
    \subfigure[CM-Het-NetS] {
		\includegraphics[width=0.25\textwidth]{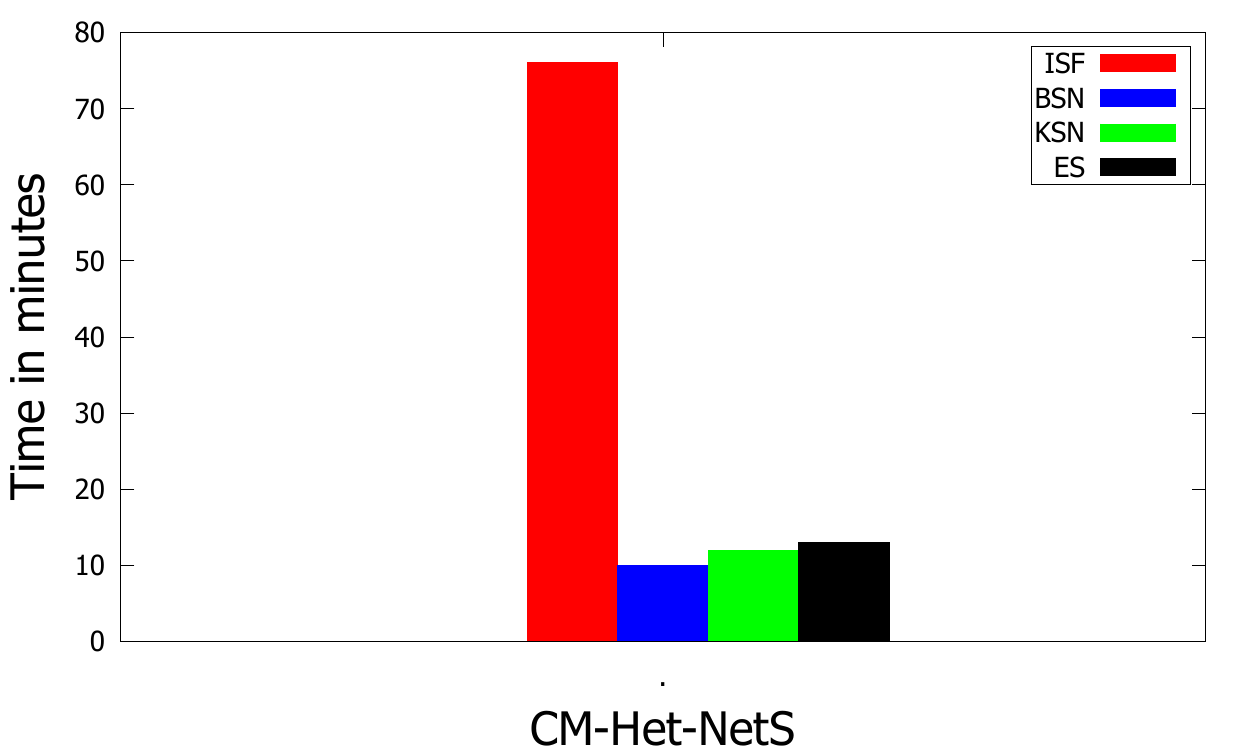}
		\label{chn_time}
    } 	
    \subfigure[Twitter-FSQ] {
		\includegraphics[width=0.25\textwidth]{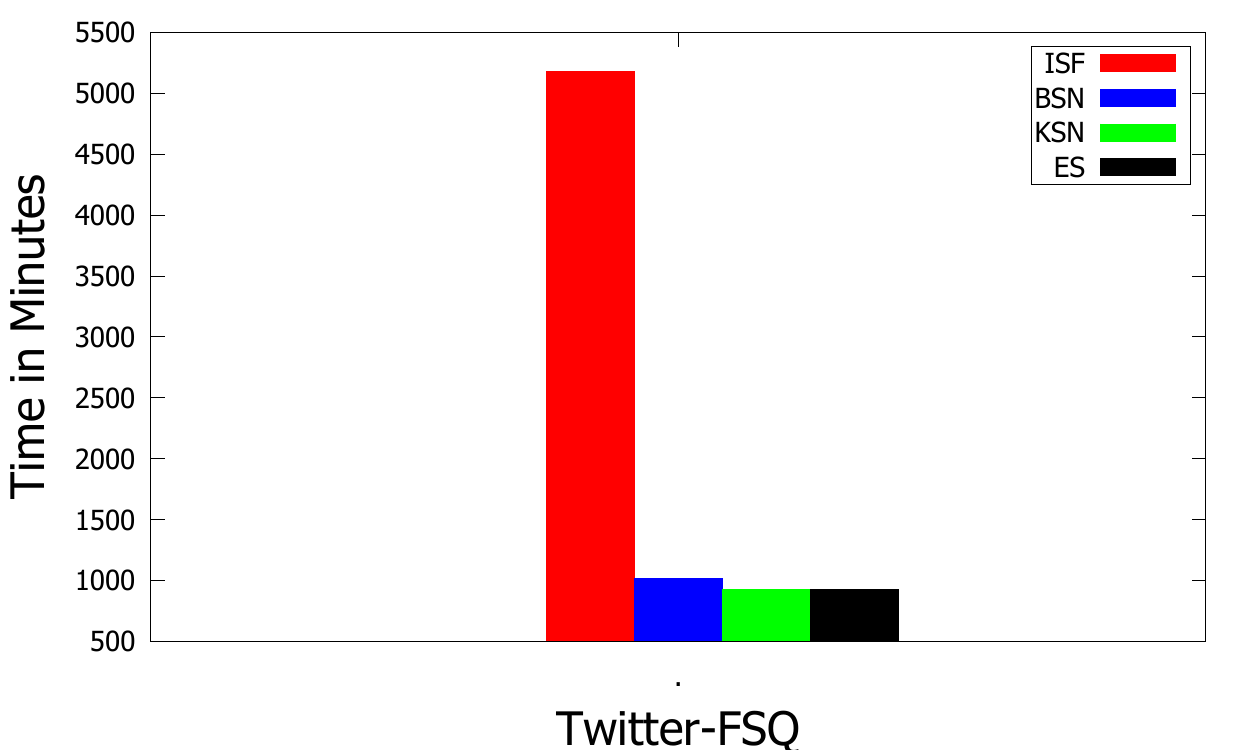}
		\label{ft_time}
    } 	

    \label{fig:run_time}
    \caption{
    (a), (b) Comparison of the four methods with no overlap and
    500 overlapping users, respectively, in the scale-free multiplex setting.
    (c)--(e) Running time comparison for the four algorithms on different multiplexes.}
    \vspace{-12pt}
\end{figure*}

\subsubsection{Algorithm performance on synthesized multiplex networks}

The performance of all four algorithms on the scale-free
multiplex $\mathscr{H}_0$ with no overlapping users
is shown in Fig. \ref{fig:sf_syn_0ov}.
As may be suggested by the analysis of the performance ratio
for KSN, which depends on the number of overlapping users
$o = 0$, the performance of KSN equals
or exceeds ISF. ES, where seeds are split evenly among
the three layers, requires 40 seed users to get
comparable activation to KSN and ISF at 20 seed users.
BSN does the worst since it is choosing seed users from a single layer --
since all three layers have 1000 nodes, BSN 
cannot activate more than $33\%$ of the multiplex,
which it approaches as the number of seed nodes $l \ge 35$.

Next, we considered the $\mathscr{H}_{500}$, the 
multiplex with the same layers but 500 overlapping users,
which is 1/6 of the original 3000 users. The performance
is shown in Fig. \ref{fig:sf_syn_500ov}. In the case
of this significant overlap, ISF outperforms KSN.
BSN is no longer limited to activation of at most
$33\%$ and performs similarly to ES.

These results on the synthesized scale-free multiplex
demonstrate how, in the case of small overlap we
expect KSN to perform as well or better than ISF;
however, as overlap increases, the performance
of KSN will degrade with respect to ISF, a 
statement that we have demonstrated theoretically
in section \ref{KSN-perf}.
\subsubsection{Running time}
In Fig. \ref{sf_time}, we compare the running time for the four algorithms on
$\mathscr{H}_0$ and $\mathscr{H}_{500}$.  
The effect of parallelization by layer of the
multiplex on the running time may be seen; 
note that for KSN, BSN, and ES, we are
using CELF++ on each layer.  This algorithm could be 
replaced
by a more efficient single layer algorithm, which would further improve running time as compared to ISF.  

\subsubsection{Role of overlapping users}
\paragraph{On total activation}
To investigate the role of overlapping users further, we varied the number of overlapping users from 50 to 400 in
the ER multiplex setting. 
The effect of this on the total activation of ISF
can be seen in Fig \ref{fig:er_syn_ov_inf_spread_heat}. Increasing the number of overlapping users increases the number of influenced users. 

\paragraph{On performance of KSN}
We have already noticed, both from theoretical
and experimental standpoints, that the performance
of KSN with respect to the optimal solution is
expected to degrade as the number of overlapping users
increases.  In this section, we examine how the
performance of KSN compares with itself on the
scale-free synthesized multiplexes as the
number of overlapping users increases.
The results are shown in Fig. \ref{fig:sf_ksn_ov}.  
As the number of 
overlapping users increases, the algorithm's performance improves drastically, demonstrating
the efficacy of KSN even when overlapping users 
exist. With the overlapping percentage at 16.67\%, 
KSN activates over 80\% of the nodes in the
multiplex with just five seed nodes, as opposed to in the case of no overlap,
where activation is at roughly 35\% with 5 seed nodes.  
In addition, this experiment provides further evidence of the strong benefit overlapping users provide in
the influence propagation.

\begin{figure}
  \centering  
	\subfigure[Total activation of ISF: Overlap Size vs Seed Size (Erdos-Renyi)] {
		\includegraphics[scale=0.27]{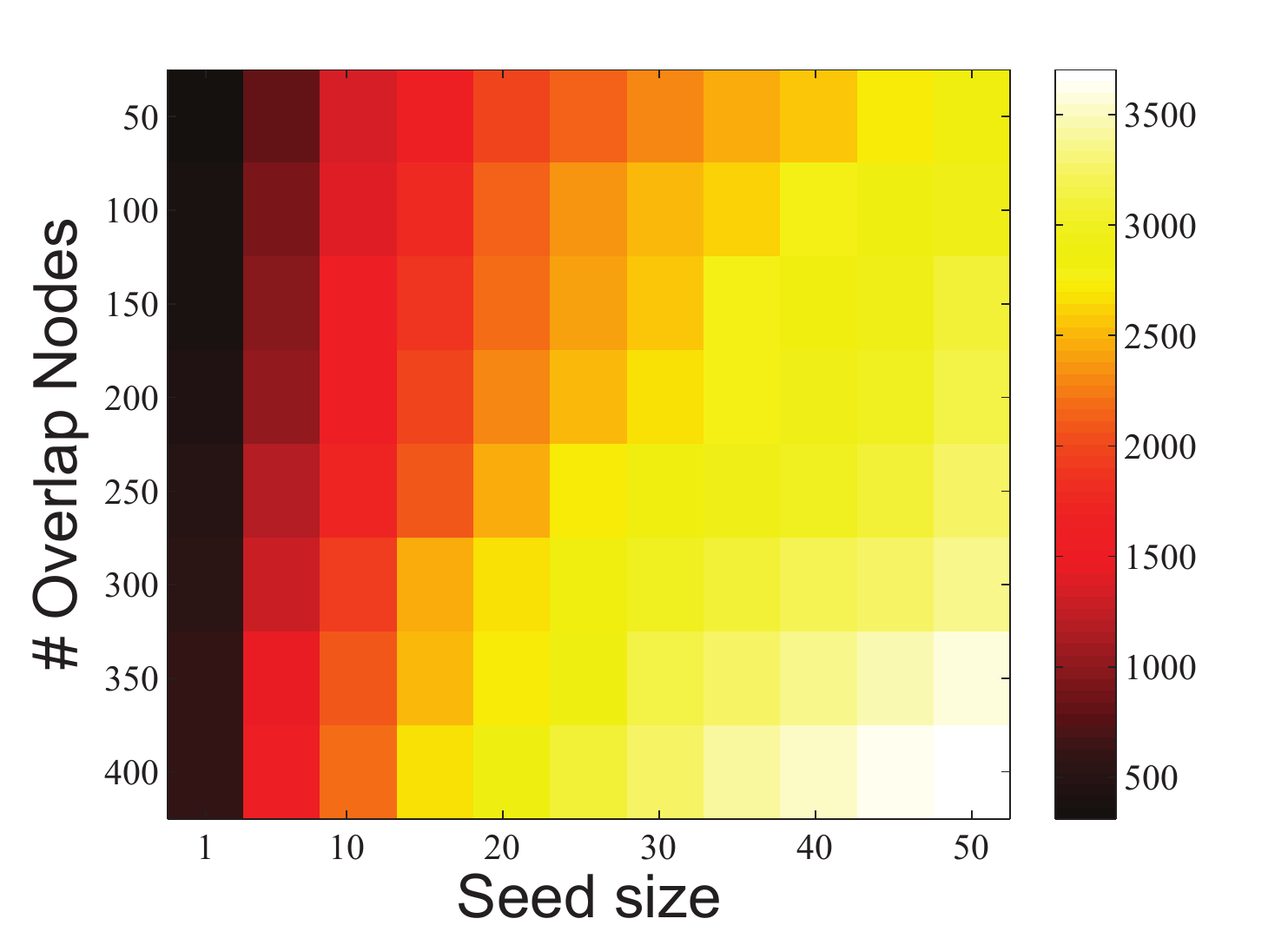}
		\label{fig:er_syn_ov_inf_spread_heat}
	} 	
	\subfigure[Effect of overlapping on KSN performance (Scale-free)] {
		\includegraphics[height=0.12\textheight]{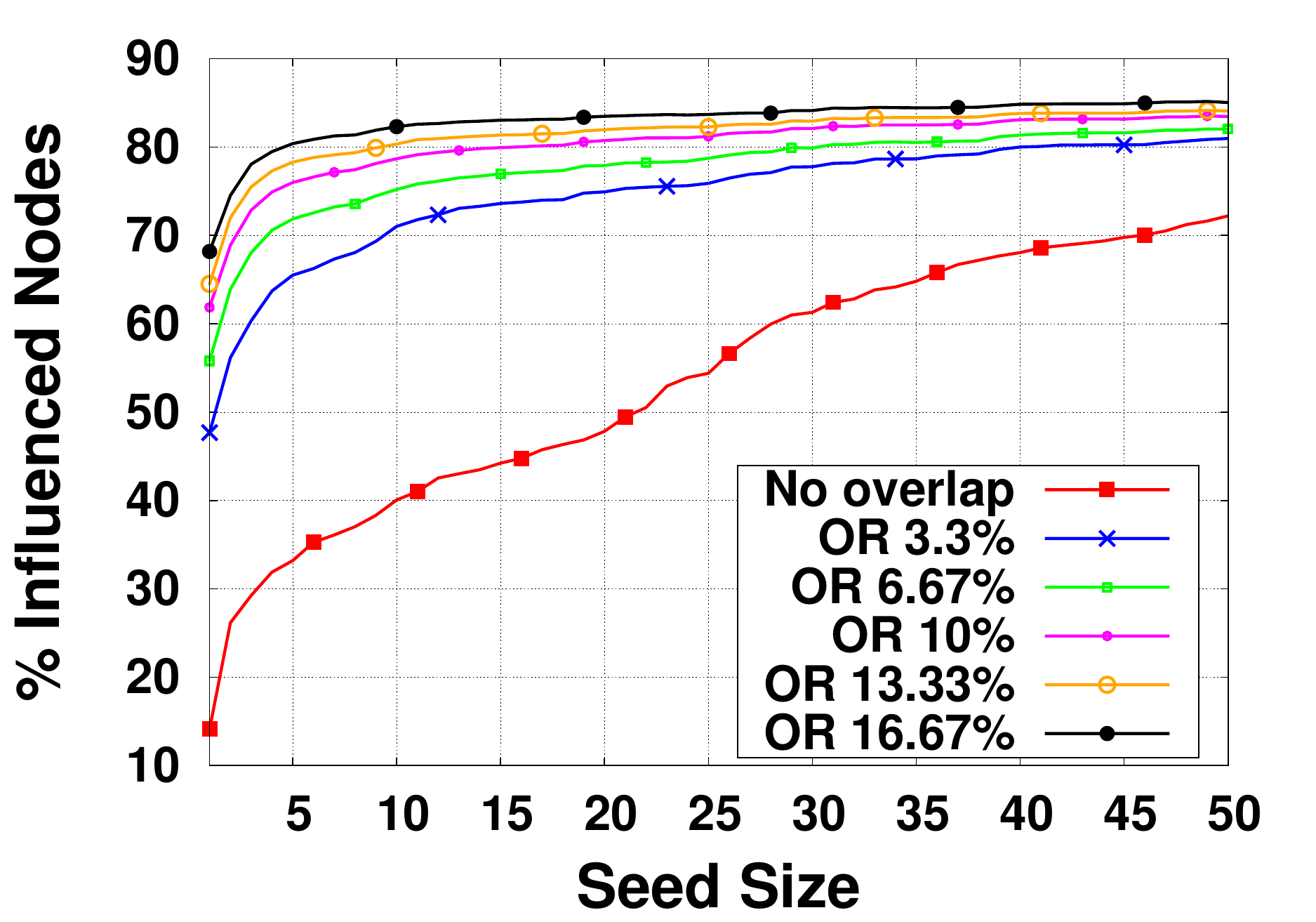}
		\label{fig:sf_ksn_ov}
	}
	  \caption{Influence Spread for different overlap ratio, Diffusion Models- Net1: IC, Net2: LT}
			\label{fig:er_syn_ov_dist}
			\vspace{-12pt}
\end{figure}

\subsection{Algorithm performance on real networks} \label{sect:real}
The first real multiplex we consider is based
upon sections of \emph{Twitter} and \emph{Foursquare} (\emph{FSQ}) networks. The generation of this dataset
is described in \cite{shen2012interest}; overlapping
users were identified by using Foursquare API v1 to
identify the Twitter usernames corresponding to a 
Foursquare account. The weight of each link
in Twitter is inferred by using frequency
of tweets between users.  In Foursquare, the weight
of each link is assigned value 1, due to the lack
of a message dataset. The number of overlapping nodes is 4100, see Table \ref{tab:DataSet}.

The second real multiplex is based upon academic
collaboration networks, described in \cite{shen2012interest}.  The layers are organized by the research area:
\emph{Condensed Matter} (\emph{CM}), \emph{ High-Energy Theory} (\emph{Het}) \cite{newman2001structure}, and  \emph{Network Science} (\emph{NetS}) \cite{newman2006finding}. 
A user is considered to be overlapping if he or she 
has published in two or more of these three fields.
The CM-HET-NetS (CHN) multiplex is considered 
an undirected network throughout the experiments.  
The number of 
overlapping users are 
2860, 517, and 90 between CM-HET, CM-NetS, and HET-NetS, respectively; 75 users are present in all three networks.

\subsubsection{Model selection}
Saito et al. \cite{saito2010selecting} 
have performed machine-learning
techniques to match variants of the IC and LT
models with real propagation events.
They found that even on the same network,
different propagation 
events may be better explained by disparate
models. 
In all the Twitter-FSQ experiments we assigned 
Twitter the LT model, and Foursquare the IC model; 
experiments that swapped the
model selection gave similar qualitative results.  
On CHN, we assigned CM, HET and NetS the
LT, IC and LT models, respectively. Thresholds
were assigned uniformly
randomly in $[0,1]$, while edge weights were
determined as described above for Twitter, 
and randomly in $[0,1]$ for the other networks.
Notice that with this choice of models,
Theorem \ref{theorem1} applies; 
therefore, the approximation ratios
of ISF and KSN hold.


\begin{figure*}
  \centering  
	\subfigure[Total activation in Twitter-FSQ] {
		\includegraphics[scale=0.25]{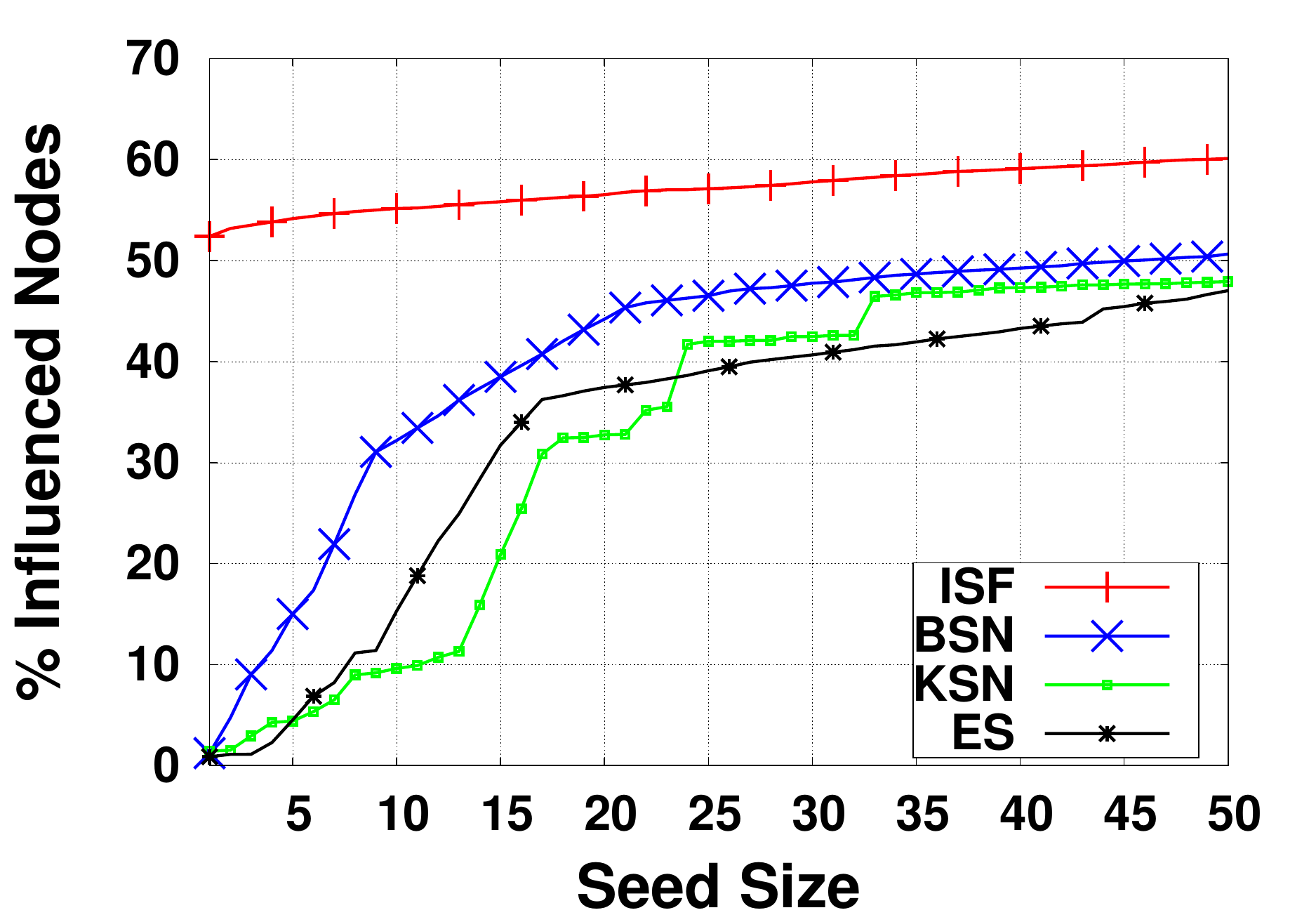}
		\label{fig:ft_inf_spread_isf}
	} 
	\subfigure[Activated Node Composition (Twitter-FSQ)] {
		\includegraphics[scale=0.25]{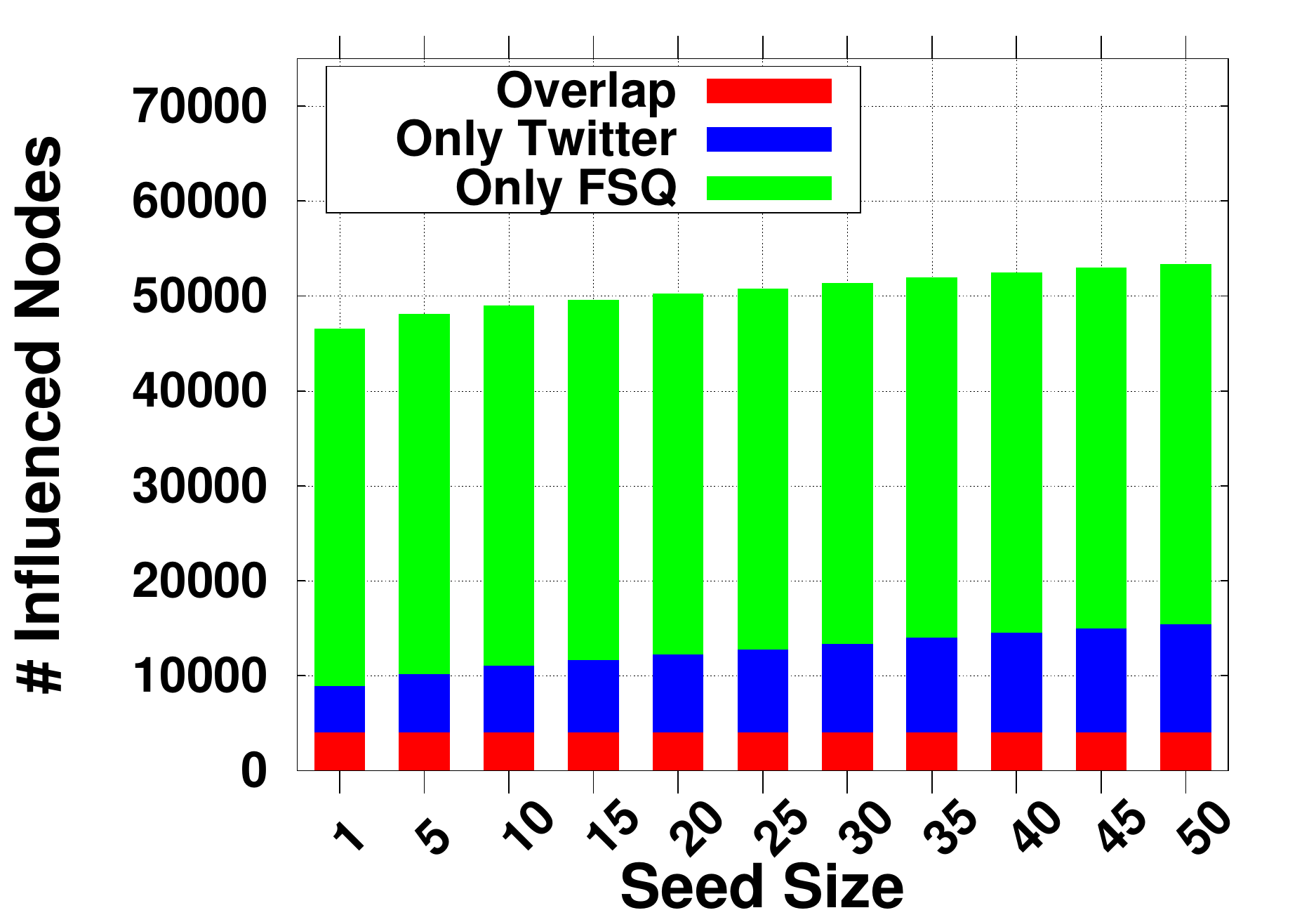}
		\label{fig:ft_inf_dist}
	}
	\subfigure[Seed Node Composition (Twitter-FSQ)] {
		\includegraphics[scale=0.25]{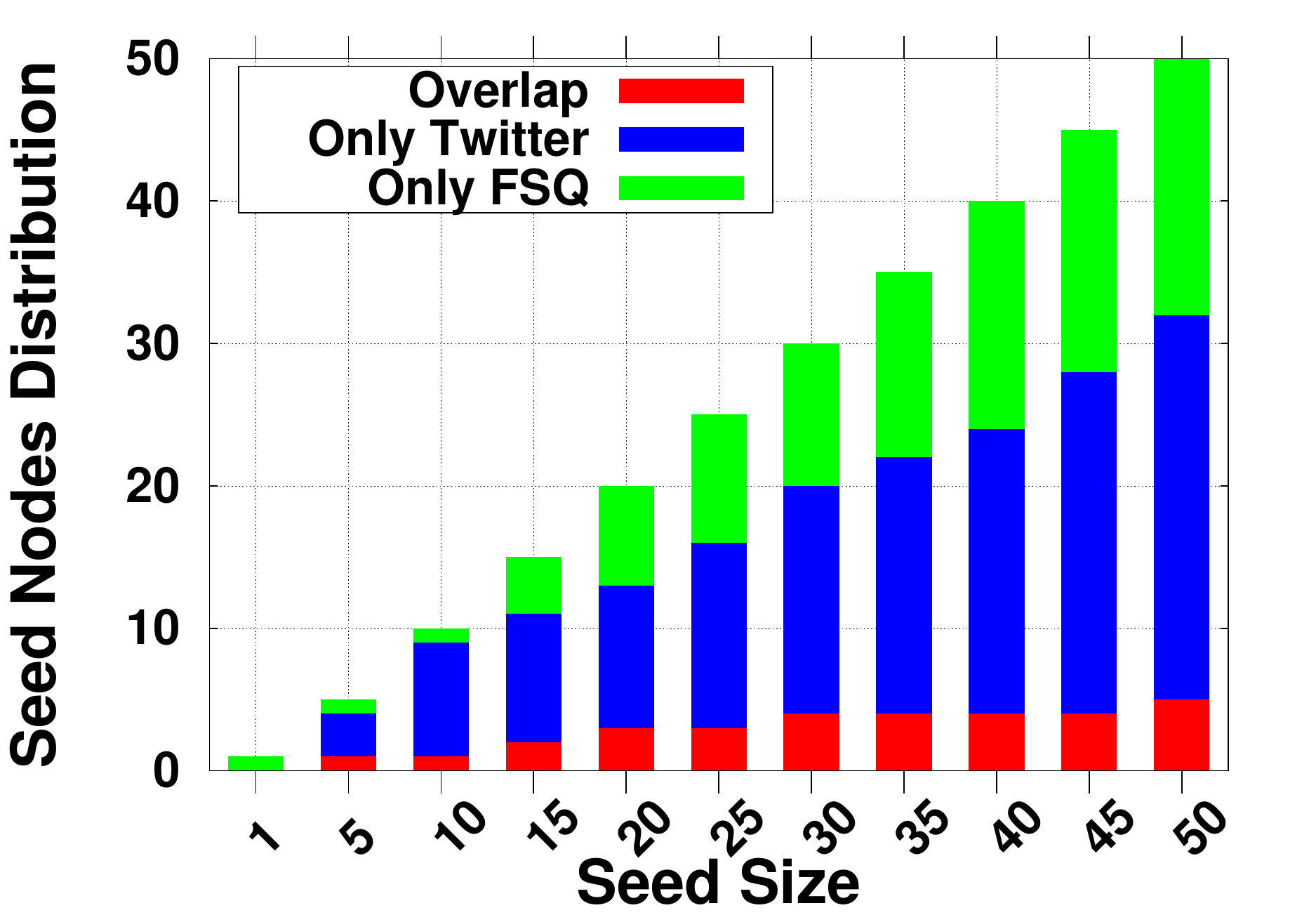}
		\label{fig:ft_seed_dist}
	}
	\subfigure[Total activataion in CM-Het-NetS] {
		\includegraphics[scale=0.25]{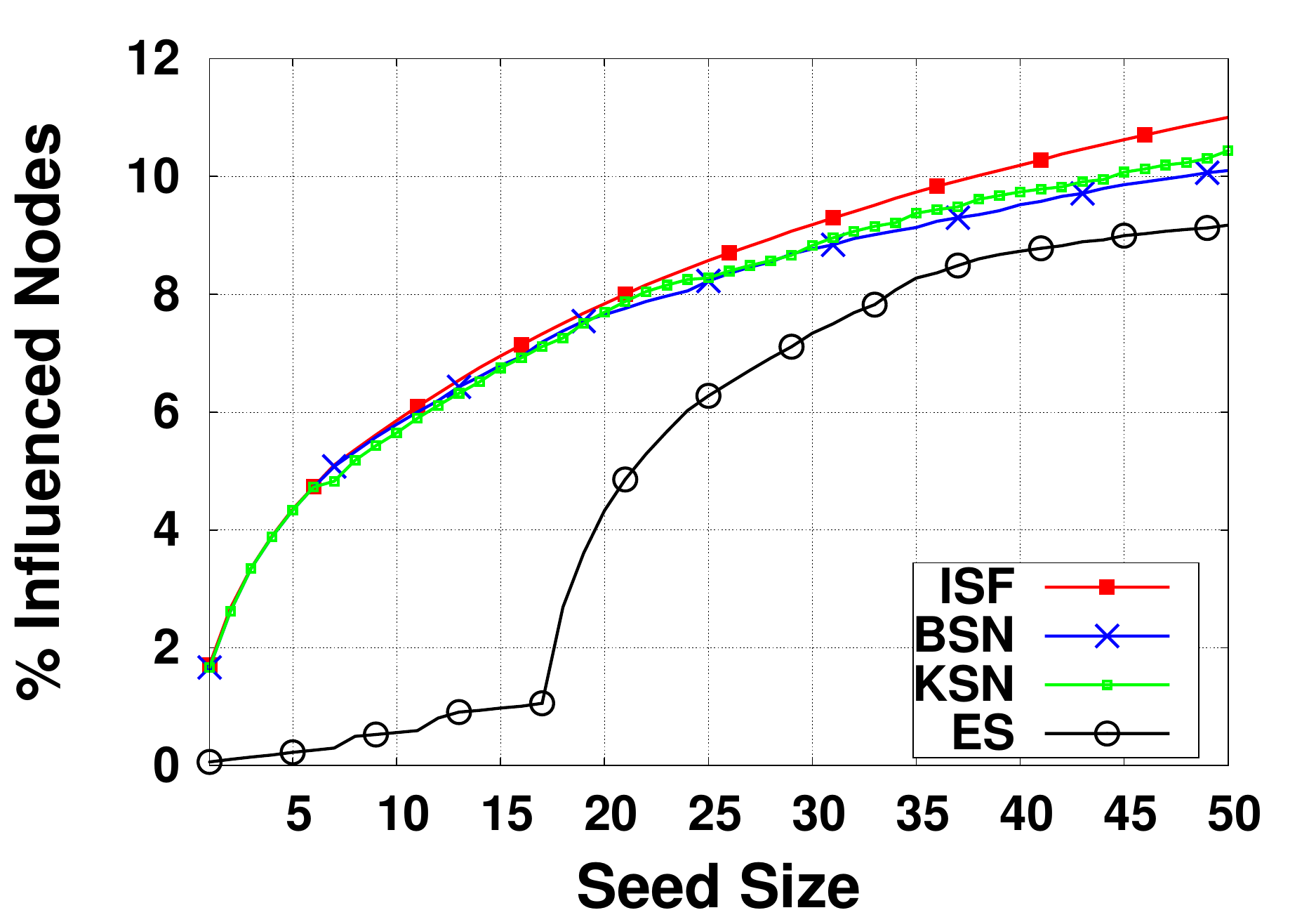}
		\label{fig:chn_inf_spread_isf}
	} 
	\subfigure[Activated Node Composition (CHN)] {
		\includegraphics[scale=0.25]{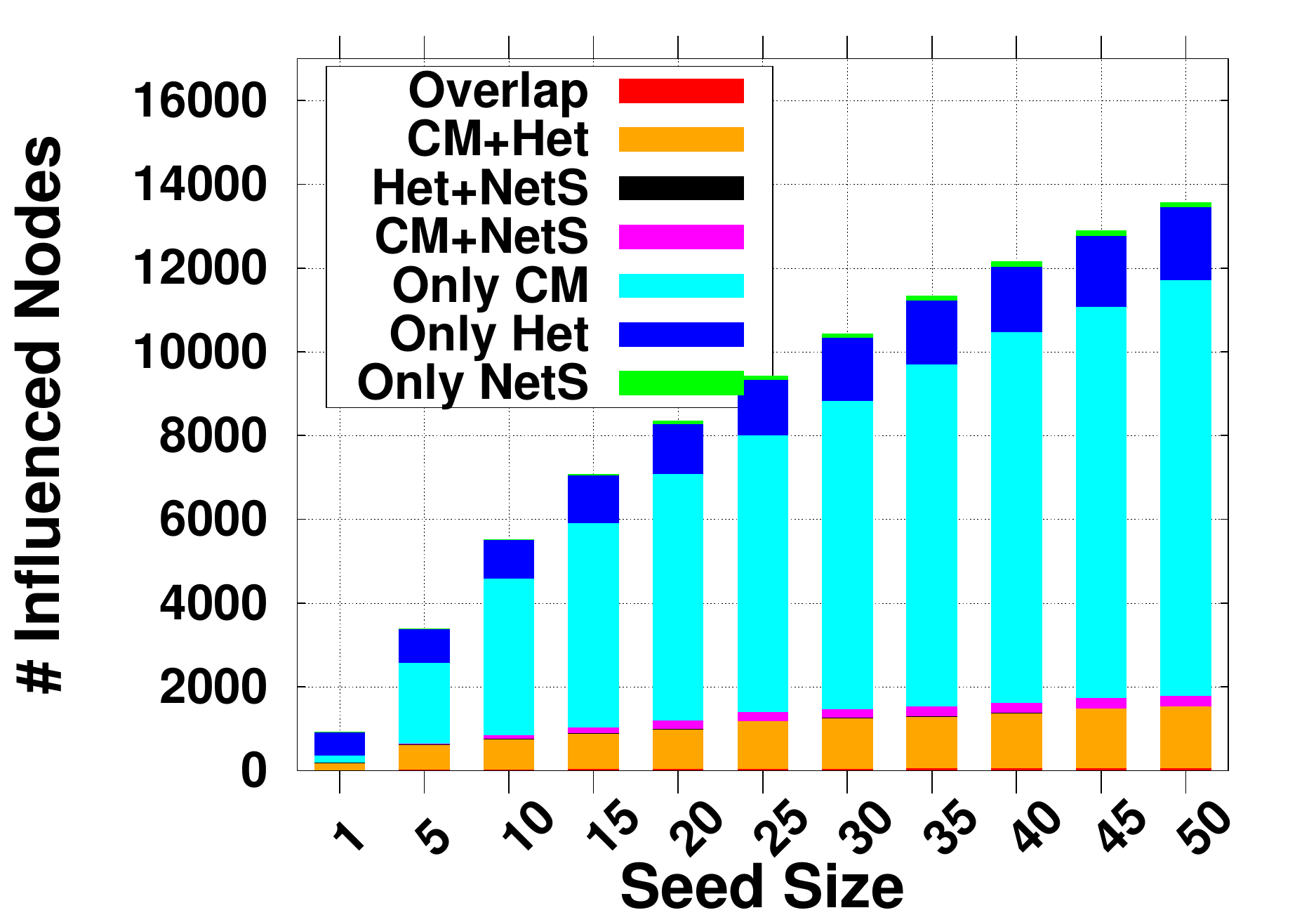}
		\label{fig:chn_inf_dist}
	}
	\subfigure[Seed Node Composition (CHN)] {
		\includegraphics[scale=0.25]{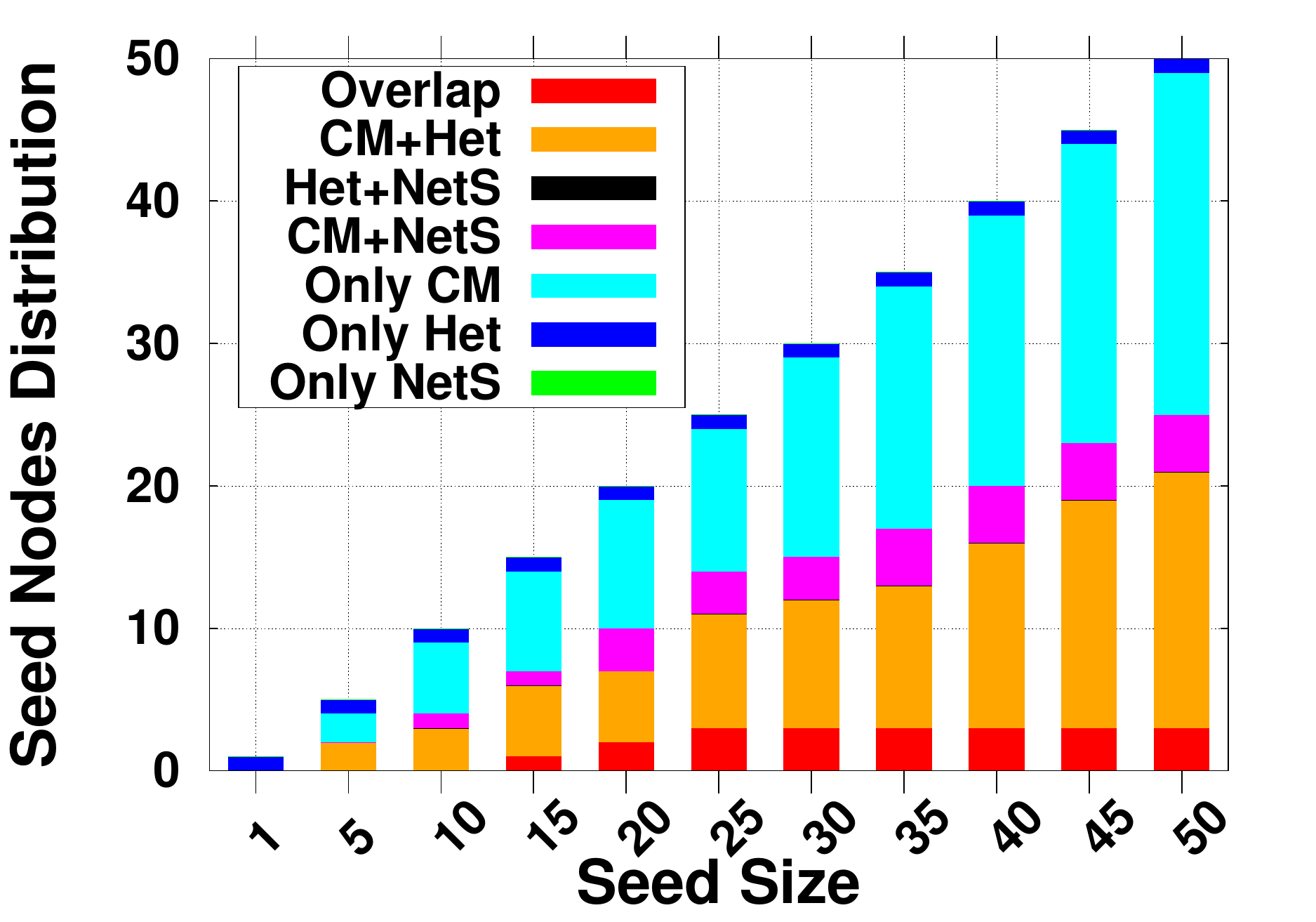}
		\label{fig:chn_seed_dist}
	}	
  \caption{Twitter-FSQ Network (top figures), co-author networks (bottom figures)  }
			\label{fig:ft_chn_isf}
			\vspace{-10pt}
\end{figure*}

\begin{table}
	\centering
    \begin{tabular}{|c|c|c|c|c|c|}
        \hline
        Networks        & Nodes & Edges    &Avg Deg \\ \hline
        Twitter    & 48277 & 16304712      & 289.7\\ \hline
        FSQ & 44992 & 1664402  &       35.99\\ \hline
				CM & 40420 & 175692        & 8.69\\ \hline
				Het & 8360 & 15751      & 1.88\\ \hline
				NetS & 1588 & 2742      & 1.73\\ \hline
    \end{tabular}
		\vspace{0.05in}
	\caption{Traces of real networks}
	\label{tab:DataSet}
\vspace{-20pt}
\end{table}

It is evident from Fig. \ref{fig:ft_inf_spread_isf} 
that the seeds found by ISF in Twitter-FSQ 
network outperforms BSN (20\% larger for $l=50$) as well as ES and KSN. An interesting observation in this figure is that relatively few (overlapping) nodes are responsible for a lot of the propagation in the ISF case -- the seed node composition is shown in Fig. \ref{fig:ft_seed_dist}. Also, in Fig. \ref{fig:ft_inf_dist}, we see influence spread is larger in FSQ compared to Twitter. 

Since the Condensed Matter Network is comparatively larger than the other two, most of the seed users are selected from it as shown in Fig. \ref{fig:chn_seed_dist}. Therefore, a significant number of finally activated users also reside in this network. The influence spread obtained in the multiplex network only taking the seeds of BSN and KSN are very close to that obtained by the seed nodes identified by ISF. Nonetheless, ISF outperforms them as the seed set becomes larger, as shown in Fig. \ref{fig:chn_inf_spread_isf}; this
again illustrates the benefit of taking into account overlapping users in the solution.

As depicted in Fig. \ref{fig:ft_inf_dist}, the composition of influenced users in the Twitter-FSQ network suggests that majority of influenced users in the multiplex network belongs to FSQ which implies propagation spreads easily in this network compared to Twitter. The same observation can be drawn for CM network in case of co-author network, shown in Fig. \ref{fig:chn_inf_dist}. As illustrated in Fig. \ref{fig:ft_seed_dist} and Fig. \ref{fig:chn_seed_dist}, the seed set of the multiplex network identified by ISF contains a much higher number of nodes from a specific network than other networks. 
In addition, overlapping users show significant role in diffusing information by occupying a considerable fraction of the seed set chosen by ISF.

\emph{Running time:}
In Fig. \ref{fig:run_time}, we compare the running time for the four algorithms on
CHN, and Twitter-FSQ multiplexes.  
The effect of parallelization by layer on
the running time may be seen, with ISF taking much longer
in all cases.

\subsection{Scalability of KSN} \label{sect:ksn-scale}
\begin{figure}
\centering
\subfigure[$n_{ER}=10^5$] {
  \includegraphics[scale=0.20]{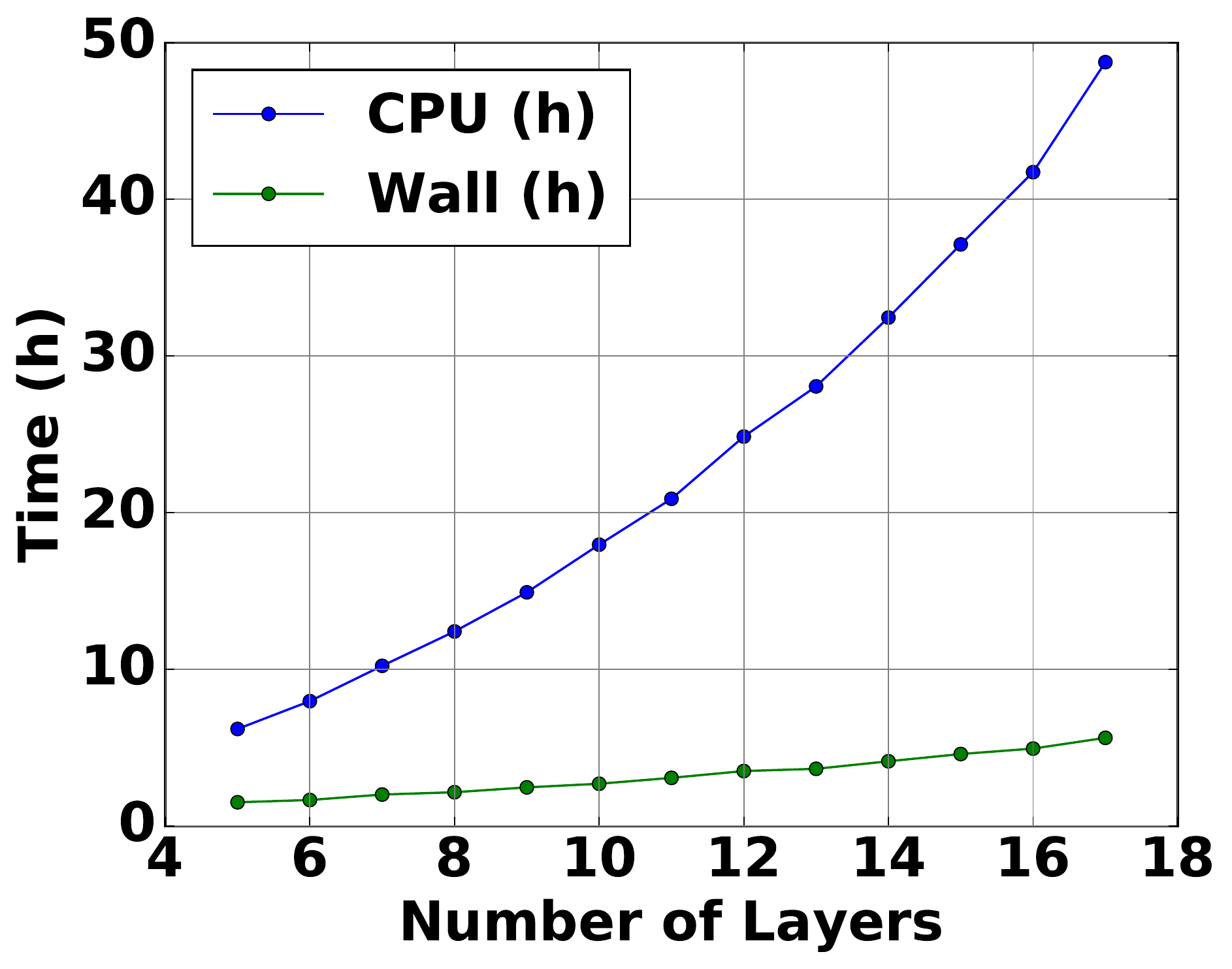}
  \label{fig:ksn-scale-k}
}
\subfigure[$k = 10$] {
  \includegraphics[scale=0.20]{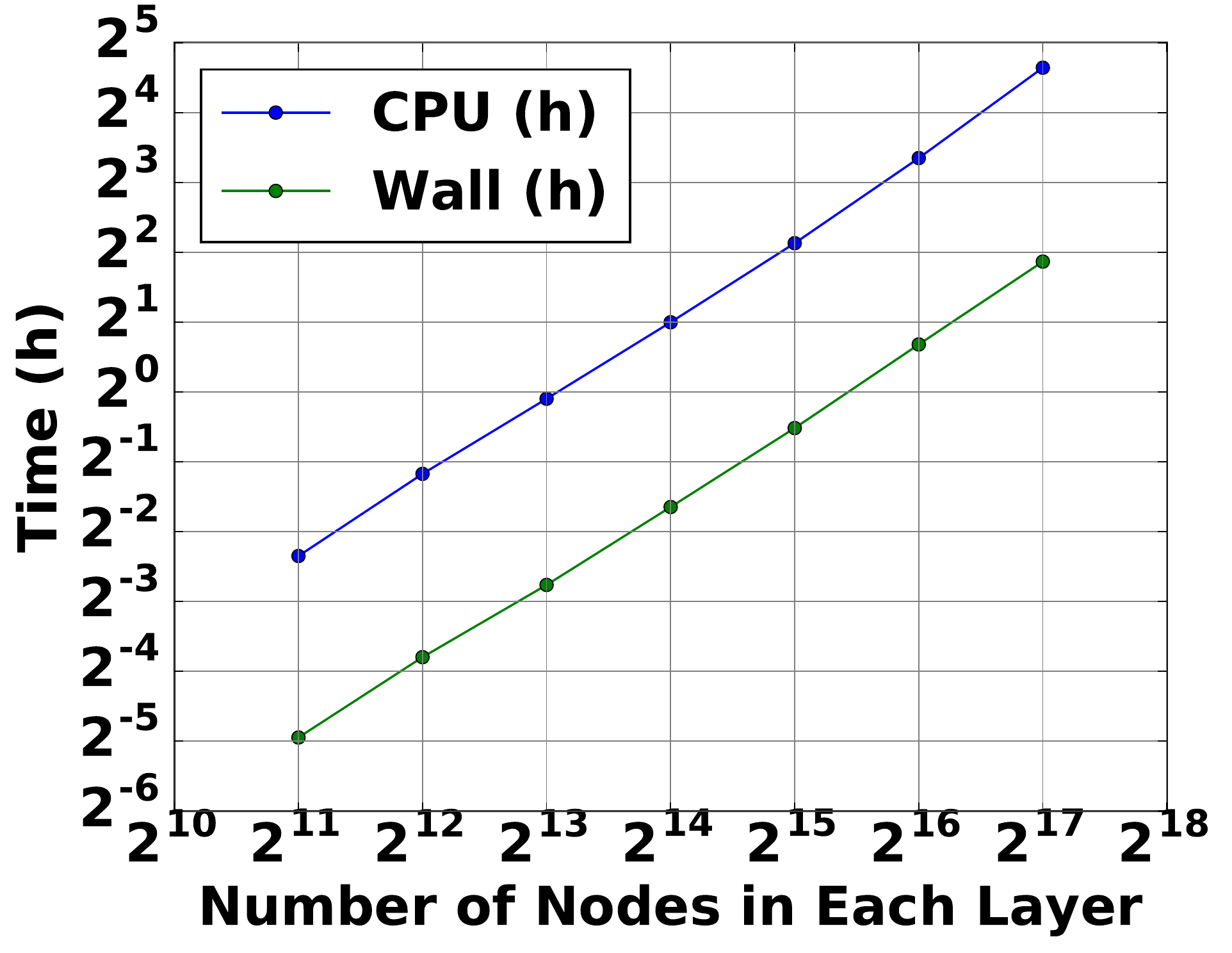}
  \label{fig:ksn-scale-n}
}
\caption{Scalability of KSN} \label{fig:ksn-scale}
\end{figure}
In our final set of experiments, we demonstrate
the scalability of KSN on large multiplexes. 
For these experiments, we used synthesized multiplexes,
where each layer is an ER network with average degree 5.
Each layer has the same number $n_{ER}$ of vertices. 
Layer $i$
is assigned model IC if $i$ is even and model LT otherwise,
with edge weights uniformly chosen in $(0, 0.1)$. The number of 
overlapping users is set to $o = 0.1 n_{ER}$, and 
each overlapping user is present in all layers of the multiplex. For these 
experiments, the number of seeds $l = 100$. Since the solution
of IMM of size $m$ is not necessarily contained in the solution
of size $m + 1$, IMM is run on each layer for all values from 1 to $l$,
as indicated in the pseudocode of KSN.

Results for the running time of KSN are shown in Fig.
\ref{fig:ksn-scale}. The first experiment varies the
number $k$ of layers in the multiplex from 5 to 17, with $n_{ER} = 10^5$.
Thus, the largest multiplex in this experiment has
$1.54 \times 10^6$ unique users. In Fig. \ref{fig:ksn-scale-k},
we show the running time in hours of KSN versus $k$, for both
total CPU time and the wall-clock time. Thus, on the
largest multiplex with $k = 17$,
KSN finished in roughly 5 hours of wall-clock time, demonstrating a high 
level of parallelization. Results for the second
experiment are shown in Fig. \ref{fig:ksn-scale-n},
where the number of layers is fixed at $k = 10$,
and $n_{ER}$ is varied from 2048 to 131072. In this experiment,
both the wall-clock time and CPU time increase linearly,
with the wall-clock time below 4 hours on the largest
multiplex with $1.19 \times 10^6$ unique users.

\section{Related Works}\label{label:relatedwork}
In a seminal work on single-layer networks,
Kempe et al. showed the IC and 
LT models were submodular \cite{Kempe2003} by utilizing
a ``live edge'' approach, thereby allowing the use of
a greedy algorithm to approximate the influence maximization
on single networks.  The ``live edge'' approach implicitly
establishes the stronger GDS property as defined above for
these models.  The inapproximability of the classical max
coverage problem precludes any better approximation to the
influence maximization problem.  Since this work,
there have been a number of improvements to the running time
of the greedy algorithm on single networks. The first improvement
was through the use of priority queue to do a ``lazy evaluation''
of estimated influence
in \cite{leskovec07}.  
Cohen et al. provided an approximation algorithm for influence maximization one or two orders of magnitude faster than previous works on single-layer 
networks \cite{cohen2014sketch}; Kuhnle et al. \cite{Kuhnle2017}
adapted this framework for the threshold activation problem.
Borgs et al. provide a nearly runtime-optimal algorithm on single networks with the IC model \cite{Borgs2014max}, and Tang et al. \cite{Tang2014} provided a fast algorithm with with high probability achieves 
the greedy ratio $1 - 1/e - \epsilon$ for
the triggering model; this method was further improved by Tang \cite{Tang2015} 
using martingales. Nguyen et al. \cite{Nguyen2016,Nguyen2016a} and Huang et al. \cite{Huang2017} have
further improved the sampling techniques to yield even faster algorithms for the single-layer network
problem with the $1-1/e-\epsilon$ ratio; also, Li et al. \cite{Li2017} have developed
a scalable and nearly optimal algorithm.

A considerable number of works have studied influence maximization 
for variants of IC models and its extensions such as 
\cite{Chen2010,Kempe2005}. For a deterministic variant of
LT, Feng et al. \cite{Feng2009} showed NP-completeness for 
the problem and Dinh et al \cite{Dinh2012} proved the 
inapproximability as well as proposed efficient algorithms for 
this problem on a special case of LT model. In their model, the 
influence between users is uniform and a user is influenced if a 
certain fraction $\rho$ of his friends are active.

Researchers 
have started to explore the influence maximization problem
on multiplex networks with works of 
Yagan et al. \cite{Yagan2012} and Liu et al. \cite{Liu2012} which 
studied the connection between offline and online networks. The first 
work investigated the outbreak of information using the SIR model 
on random networks. The second one analyzed networks formed by 
online interactions and offline events. The authors focused on 
understanding the flow of information and network clustering but 
not solving the heterogeneous influence problem. 

Shen et al. \cite{shen2012interest} explored the information propagation 
in multiplex online social networks taking into account the interest and
 engagement of users. The authors combined all networks into one network
 by representing an overlapping user as a super node. This method cannot
 preserve the heterogeneity of the layers.  Nguyen et al. 
\cite{nguyen2013ICDM} studied the influence maximization problem which 
handles multiple networks but only considers homogeneous diffusion process
 across all the networks. None of these works took into consideration 
 the heterogeneity of diffusion processes in multiplex networks.
On the other hand, our scheme overcomes these shortcomings by enabling different networks to have different influence propagation models. Pan et al. \cite{Pan2017} study 
threshold activation problems on multiplex networks
under a diffusion model with continuous time.

 
\section{Conclusion}\label{label:conclusion}
We formulate the Multiplex Influence Maximization problem that 
seeks to maximize influence propagation in a multiplex with
overlapping users and heterogeneous propagation.
We provide a property (GDS) which carries over from
single layer to multiplex propagation, giving the $1 - 1/e$ ratio for the
greedy algorithm ISF if it is satisfied on each layer. 
We also develope an approximation algorithm
KSN
that benefits from the optimizations that
influence maximization in a single network has undergone (e.g. in
\cite{Tang2014, Borgs2014max, cohen2014sketch}). We prove the
approximation ratio of KSN, which depends on the number
of overlapping users. 
As demonstrated in our experimental section, the performance 
KSN may fall short of ISF when overlapping effects become large.  
Due to the long
running time of ISF, future work would include 
attempting to find faster approximation 
algorithms in the heterogeneous multiplex setting.
\section{Acknowledgements}
This work is supported in part by NSF grant CNS-1443905 and
DTRA grant HDTRA1-14-1-0055.

\bibliographystyle{unsrt}
\bibliography{ref,mendeley}
\end{document}